\newcommand\NOTYETSOUMISSIONLICS[1]{}
\newcommand\PREUVEPASENANNEXE{}

\newcommand\CORPSPREUVEDE[1]{}
\newcommand\ANNEXEPREUVEDE[2]{
	\subsection{Proof of #1}
	#2
}
\providecommand\PREUVEPASENANNEXE[1]{#1}
\PREUVEPASENANNEXE{
	\renewcommand\CORPSPREUVEDE[1]{#1}
	\renewcommand\ANNEXEPREUVEDE[2]{}
}

\newcommand\SANSCOMMENTAIRE[1]{}
\newcommand\ANONYME[1]{#1}
\NOTYETSOUMISSIONLICS{
\renewcommand\SANSCOMMENTAIRE[1]{#1}
\renewcommand\ANONYME[1]{#1}
}

\documentclass[conference,10pt]{IEEEtran}
\newcommand\NOLICSSTYLE[1]{#1}
\usepackage{hyperref}

\usepackage{mathrsfs}  

\usepackage{color}

\usepackage{todonotes}

\newcommand{\olivier}[2][]{\SANSCOMMENTAIRE{\todo[inline,color=blue!40,caption={2do}, #1]{\begin{minipage}{\textwidth-4pt}
Olivier:			#2\end{minipage}}}}
\newcommand{\olivierpasimportant}[2][]{\SANSCOMMENTAIRE{\todo[inline,color=blue!10,caption={2do}, #1]{\begin{minipage}{\textwidth-4pt}
\tiny Olivier:			#2\end{minipage}}}}
\newcommand{\olivierplan}[2][]{%\SANSCOMMENTAIRE{\todo[inline,color=blue!10,caption={2do}, #1]{\begin{minipage}{\textwidth-4pt}
}
\newcommand{\olivierplusimportant}[2][]{}

\newcommand{\olivierpourmanon}[2][]{\SANSCOMMENTAIRE{\todo[inline,color=blue!40,caption={2do}, #1]{\begin{minipage}{\textwidth-4pt}
Olivier $\to$ Manon:			#2\end{minipage}}}}
\newcommand{\olivierpourmanoninfo}[2][]{\SANSCOMMENTAIRE{\todo[inline,color=red!20,caption={2do}, #1]{\begin{minipage}{\textwidth-4pt}
\small Olivier $\to$ Manon:   *POUR INFO*	 		#2\end{minipage}}}}

\newcommand\mylabelbt[1]{\cite[#1]{brattka2008tutorial}}
\newcommand\mylabelwei[1]{\cite[#1]{Wei00}}

\NOLICSSTYLE{
\ifCLASSINFOpdf
\else
\fi
}

\usepackage{amsmath}

\usepackage{amsthm}
\usepackage{amsfonts}
\usepackage{amssymb}
\newtheorem{theorem}{Theorem}
\newtheorem{lemma}[theorem]{Lemma}
\newtheorem{proposition}[theorem]{Proposition}%[lemma]
\newtheorem{corollary}[theorem]{Corollary}%[lemma]
\newtheorem{definition}[theorem]{Definition}%[lemma]
\newtheorem{remark}[theorem]{Remark}
\newtheorem{theoremd}[theorem]{Theorem}
\newtheorem{lemmad}[theorem]{Lemma}
\newtheorem{propositiond}[theorem]{Proposition}%[lemma]
\newtheorem{corollaryd}[theorem]{Corollary}%[lemma]

\newcommand\COULEUR[1]{#1}
\SANSCOMMENTAIRE{\renewcommand\COULEUR[1]{\textcolor{blue}{#1}}}

\newtheorem{theoremaprouver}[theorem]{\COULEUR{Conjecture/Theorem to be proved}}%[theorem]
\newtheorem{corollarymanon}[theorem]{\COULEUR{Corollary}}%[theorem]
\newenvironment{proofmanon}{\begin{proof}}{\end{proof}}
 \newcommand\R{\mathbb{R}}
\newcommand\Q{\mathbb{Q}}
\newcommand\N{\mathbb{N}}

\newcommand\vectorl[1]{{\mathbf#1}}
\newcommand\tu[1]{\vectorl{#1}}

\newcommand\vx{\vectorl{x}}

\newcommand\vy{\vectorl{y}}
\newcommand\vz{\vectorl{z}}

\newcommand\mG{\mathcal{G}}
\newcommand\mP{\mathcal{P}}
\newcommand\mL{\mathcal{L}}
\newcommand\mF{\mathcal{F}}

\newcommand{\cp}[1]{\operatorname{#1}}
\newcommand{\myop}[1]{\operatorname{#1}}

\newcommand\DTIME{\cp{DTIME}}
\newcommand{\Ptime}{\cp{PTIME}}      % P
\newcommand\PTIME{\Ptime}

\newcommand{\FP}{\ensuremath{\operatorname{FP}}}
 \newcommand{\PSPACE}{\ensuremath{\operatorname{PSPACE}}}
  \newcommand{\NLOGSPACE}{\ensuremath{\operatorname{NLOGSPACE}}}
    \newcommand{\coNLOGSPACE}{\ensuremath{\operatorname{coNLOGSPACE}}}

  \newcommand{\NSPACE}{\ensuremath{\operatorname{NSPACE}}}
  \newcommand{\sSPACE}{\ensuremath{\operatorname{SPACE}}}

\newcommand\motnouv[1]{\emph{#1}}

\newcommand\bd{\operatorname{bd}}
\newcommand\cB{\overline{B}} %closed ball
\newcommand\length{\ell}
\newcommand\distance[2]{d(#1,#2)}
\newcommand\closure[1]{cls(#1)}

\newcommand\CONFIGMT[3]{(#1,#2,#3)}
\newcommand\CONFIGURATIONSMT{\mathcal{C}_{\M}}

\newcommand\M{\mathcal{M}}
\renewcommand\P{\mathcal{P}}
\newcommand\TIME[2]{TIME(#1,#2)}
\newcommand\EXACTTIME[2]{L^{[#2]}(#1)}
\newcommand\SPACE[2]{SPACE(#1,#2)}
\newcommand\EXACTSPACE[2]{L_{[#2]}(#1)}
\newcommand\LENGTH[2]{\mL(#1,#2)}
\newcommand\EXACTLENGTH[2]{L^{(#2)}(#1)}
\newcommand\PERTURBEDS[2]{L_{#2}(#1)}
\newcommand\PERTURBEDSPACE[2]{L_{\{#2\}}(#1)}
\newcommand\PERTURBEDT[2]{L^{#2}(#1)}
\newcommand\PERTURBEDTIME[2]{L^{\{#2\}}(#1)}

\newcommand\REACH[1]{R^{#1}}
\newcommand\REACHONE[1]{R^{1,#1}}
\newcommand\REACHEXACTIME[2]{R^{#1,[#2]}}
\newcommand\REACHPERTURBEDL[2]{R^{#1,#2}}
\newcommand\REACHPERTURBEDLENGTH[2]{R^{#1,(#2)}}
\newcommand\REACHEXACTSPACE[2]{{R^{#1}}_{[#2]}}
\newcommand\REACHPERTURBEDSPACE[2]{{R^{#1}}_{\{#2\}}}

\newcommand\REACHG{\REACH{G}}
\newcommand\REACHGm{\REACH{G_{m}}}
\newcommand\REACHP{\REACH{\mP}}
\newcommand\REACHPe{R^{\mP*}}
\newcommand\REACHPeomega{R^{\mP*}_{\omega}}
\newcommand\REACHPomega{\REACHP_{\omega}}
\newcommand\REACHPepsilon{\REACHP_{\varepsilon}}
\newcommand\REACHPn{\REACHP_{n}}

\newcommand{\LOOP}{\ensuremath{\operatorname{NOPATH}}}
\newcommand{\PATH}{\ensuremath{\operatorname{PATH}}}

\newcommand\logsize{\operatorname{log-size}}
\newcommand{\PATHONE}{\ensuremath{\operatorname{EDGE}}}
\newcommand{\VERTEX}{\mathcal{V}} 

\begin{document}
\title{Bare Demo of IEEEtran.cls\\ for IEEE Conferences}

\title{Measuring robustness of dynamical systems. Relating time and space to length and precision}

\ANONYME{
\author{\IEEEauthorblockN{Manon Blanc}\IEEEauthorblockA{Institut Polytechnique de Paris \\ 
Ecole Polytechnique, \\
LIX Laboratory \\
91128 Palaiseau Cedex\\
 France \\
 manon.blanc@lix.polytechnique.fr
}
 \and  \IEEEauthorblockN{Olivier Bournez}
\IEEEauthorblockA{Institut Polytechnique de Paris \\ 
Ecole Polytechnique, \\
LIX Laboratory \\
91128 Palaiseau Cedex\\
 France \\
 oilvier.bournez@lix.polytechnique.fr
}}
}
%\affil[1]{Ecole Polytechnique, LIX, 91128 Palaiseau Cedex, France}

\maketitle
%\olivierpasimportant{choisir un jour un vrai titre}
%\manon{tentative de vrai titre}

\begin{abstract}
	Verification of discrete time or continuous time dynamical systems over the reals is known to be undecidable. It is however known that undecidability does not hold  for various classes of systems when considering \emph{robust} systems: if robustness is defined as the fact that reachability relation is stable under infinitesimal perturbation, then their reachability relation is decidable. 
	In other words, undecidability implies sensitivity under infinitesimal perturbation, a property usually not expected in systems considered ``in practice'', and hence can be seen (somehow informally) as an artifact of the theory, that always assumes exactness. 
	In a similar vein, it is known that, while undecidability holds for logical formulas over the reals, it does not hold when considering $\delta$-undecidability:  one must determine whether a property is true, or $\delta$-far from being true.
	
	We first extend the previous statements to a theory for general (discrete time, continuous-time, and even hybrid) dynamical systems, and we relate the two approaches. We also relate robustness to some geometric properties of reachability relation. 
	
	But mainly, when a system is robust, it then makes sense to quantify at which level of perturbation. We prove that assuming robustness to polynomial perturbations on precision leads to reachability verifiable in complexity class PSPACE, and even to a characterization of this complexity class. We prove that assuming  robustness to polynomial perturbations on time or length of trajectories leads to similar statements, but with PTIME.
	
	Actually, these results can also be read in another way, in terms of computational power of analog models.   It has been recently unexpectedly shown that the length of a solution of a polynomial ordinary differential equation corresponds to a time of computation: PTIME corresponds to solutions of polynomial differential equations of polynomial length. Can we do something similar for PSPACE? How should we measure space robustly for dynamical systems? Our results argue that the answer is given by precision: space corresponds to the involved precision.
	
%%%% MANON.
%  It has been recently shown that the length of a solution of a polynomial ordinary differential equation corresponds to time of computation.
%  
%  We extend this to more general classes of dynamical systems.
%  
%  We also show that, for Turing machines as well as for dynamical systems, we can 
%  describe the space by the precision on the computation, providing an answer to 
%  the open question: 
%  
%  With the notion of dynamical systems also comes the notion of reachability, we prove here that, under some assumption, it is co-computably enumerable, even when we add a small perturbation, and even if the system is in continuous time. 
%  
%  Eventuually, we were able to show an equivalence between robustness and decidability, allowing us to show that the problem of knowing whether a graph is drawable is decidable under some hypothesis, and in which complexity.
%  
%  
  \end{abstract}

\olivierplusimportant{Utiliser macros 
\begin{itemize}
\item Nouvelle macro pour noté config de MT
  \CONFIGMT{q}{left}{right}
\item pour classes de complexité, la et partout
\item $\varphi$ pour la fonction d'abstraction partout: celle qui à $\vx$ associe son sommet dans le graphe
\item macro $\VERTEX$ pour un sommet de ce graphe.
\item $\vx$ plutot que $x$, et meme chose pour toute variable, si c'est  un vecteur.
\item recursively enumerable se dit computably enumerable en plus moderne. Changer partout.
\item éviter normes $\|$ et $\|$, plutôt utiliser $\distance{\cdot}{\cdot}$.
\item $\cdot$ plutot que . pour argument.
\item eviter $\subset$, plutôt $\subseteq$ pour enlever ambiguité stric, pas strict, car certains pensent que $\subset$ c'st strict d'expérience.
\item dans le meme genre, je propose de dire ''computable`` plutot que ``recursif``, car c'est plus moderne: mais ca peut se discuter.
\item bar vs overline
\end{itemize}
}

\renewcommand\subset{\subseteq}
%\olivier{Il y a des trucs à introduire quelquepart, et avant qu'on les utilise. Et systématiquement avec la macro dans le source de ce commentaire:
%Ca inclut:
%
%
%Fait, mais à comprendre si meilleur endroit:
%\begin{itemize}
%\item $\length(.)$: length of the binary representation.
%\end{itemize}
%}
\olivierplusimportant{Il y avait conflit de notations.
Utiliser (j'ai essayé)
\begin{itemize}
\item $t$ pour temps, indice dans trajectoire.
\item $n$ pour précision de $\REACHPepsilon$
\item $m$ pour préicision du graphe $G_{\delta}$
\end{itemize}
}

\section{Introduction}

Several research communities have studied the relations between dynamical systems and computation.  This includes studies of the unpredictability and undecidability in dynamical systems \cite{Moo91,KCG94}, questions related to the hardness of verification for discrete, continuous and the so-called hybrid systems \cite{AMP95}, questions related to the computational power of models of recurrent neural networks \cite{SS95}, or control theory questions \cite {BloTsi00,bournez2021survey}.

\paragraph*{Motivation related to verification} % of discrete, continuous and hybrid systems} 
while several undecidability results were stated for hybrid systems (such as Linear Hybrid Automata  \cite{HKPV95} or Piecewise Constant Derivative systems \cite{AMP95}), and while it was observed that some practical tools were  ``working well'' (terminating) in practice, a folklore (sometimes informal) conjecture appeared in the literature of verification or in several talks, stating that this undecidability is due to non-stability, non-robustness, sensitivity to the initial values of the systems, and that it never occurs in “real” systems. There were several attempts to formalize and to prove (or to disprove) this conjecture, including \cite{Franzle99,HR00}. Refer to \cite{asarin01perturbed,bournez2021survey} for some discussions.
%: See e.g. related work discussion of  \cite{asarin01perturbed}, or discussions in \cite{bournez2021survey}.

In this article, we are inspired by the approach of \cite{asarin01perturbed}, where several models for discrete, continuous time, and hybrid dynamical systems are considered: Turing machines, Piecewise affine maps, Linear hybrid automata and Piecewise Constant Derivative systems.  To each of these models is associated a notion of perturbed dynamics by a small $\epsilon$ (with respect  to a suitable metrics), and a perturbed reachability relation is defined as the intersection of all reachability relations obtained by $\epsilon$-perturbations, for all possible values of $\epsilon$. The authors show that for all these models, the perturbed reachability relation is co-computably enumerable (co-c.e., $\Pi_1^0$), and that any co-c.e. relation can be defined as the perturbed reachability relation of such models. A corollary  is, that if we define robustness as stability of the reachability relation under infinitesimal perturbations, then robust systems have a decidable reachability relation, and hence a decidable verification.

In a similar vein, \cite{gao2012delta} observed that some logics such as real arithmetic are decidable, but even the set of $\Sigma_{1}$ sentences in a language extending real arithmetic with the sine function is already undecidable, and so deciding satisfaction of such ``simple'' formulas is impossible in the general case. However, such theoretical negative result only refers to the problem of deciding logic formula symbolically and precisely. If a relaxed notion of correctness is considered, then verification becomes algorithmically solvable: one asks to answer true  when a given formula $\phi$ is true, and to return false when it is $\delta$-robustly wrong, namely $\phi$ is wrong, but actually some $\delta$-strengthening of $\phi$ is  false. In other words, undecidability happens only for  exact questions, whose decision is not stable by infinitesimal perturbations, while robust-satisfaction is decidable.

These approaches are  relating robustness to the decidability of their verification, or robust satisfaction to a decidability of satisfaction, i.e. a \emph{computability} question. Our purpose in the current article is to extend all this, and in particular \cite{asarin01perturbed}, to more general classes of dynamical systems, but also, mainly to discuss also \emph{complexity} issues: when is the verification of a system in $\PTIME$ or in $\PSPACE$?  How does it relate to the robustness of the system with respect to perturbations? 

Basically, when a system is robust, it makes sense to quantify at which level it is: which level of perturbation $\epsilon$ is allowed? We basically prove that assuming this perturbation is polynomial in the data leads to a characterization of $\PSPACE$. This provides ways to guarantee complexity of the verification of a system. Actually, we also show that this idea can be used to define various complexity classes, by playing with various concepts of robustness: we introduce the concept of time and length perturbation (basically trajectories cannot be too long), and we prove that this  leads to $\PTIME$, and hence leads to systems where verification becomes polynomial.

%In particular, a concrete application of our statements is to provides some criteria guaranteeing that verification is feasible in some time or space complexity bound. 

\paragraph*{Motivation related to models of computation}
another field of research is related to understanding how analog (continuous-time) models of computation compare to more classical discrete models such as Turing machines. In particular, a famous historical celebrated mathematical model is the 
General Purpose Analog Computer  (GPAC) model, introduced by Shanon in 1941 in \cite{Sha41}, as a
model for the Differential Analyzer \cite{Bus31}, a mechanical programmable
machine, on which he worked
as an operator.   It is known that functions computed by a GPAC correspond
to a (component of the) solution of a system of Ordinary
Differential Equations (ODE) with polynomial right-hand side (also called \motnouv{pODE}) \cite{Sha41}, \cite{GC03}.

Relating the computational power of this model to classical models such as the Turing machines, at the complexity level, is not a trivial task. In short, contrarily to discrete models of computation, continuous time models of computation (not only the GPAC, but many others) may exhibit the so-called ``Zeno phenomena'', where time can be arbitrarily contracted in a continuous system, thus allowing an arbitrary speed-up of the computation. This forbids to take the naive approach of using the time variable of the ODE as a well-defined measure of time complexity: see \cite{CIEChapter2007,bournez2021survey} for discussions. 

A celebrated recent breakthrough has been obtained in \cite{TheseAmauryenglish,JournalACM2017vantardise}, where it has been proved that a key idea to solve this issue is, when using pODEs (in principle this idea can also be used for others continuous models of computation) to compute a given function $f$,
the cost should be measured as a function of the length of the solution curve of the polynomial initial value computing the function $f$.
As ODEs is a kind of universal language in experimental sciences, this  breakthrough led to solve several open problems in various other contexts. 
This includes the proof of the existence of a universal (in the sense of Rubel)  ODE \cite{LMCS2018}; proof of the strong Turing-completeness of biochemical reactions \cite{CMSB17}, or  statements about the completeness of reachability problems (e.g. $\PTIME$-completeness of bounded reachability) for ODEs \cite{JournalACM2017vantardise}.

While it is now known that time should  be measured by length, the question on how space should be measured is unclear. How should be measure the ``memory'' used by some %analog continuous-time model of computation? Or some 
ordinary differential equation? % or some ODE? 
Can we provide a characterization of $\PSPACE$ using ODEs similar to the characterization of $\PTIME$ obtained in \cite{JournalACM2017vantardise}?  If time corresponds to length, what space corresponds to? We give some arguments to state that  basically, over a compact domain space  corresponds to precision, while it corresponds to the log of the size of some graph for systems over more general domains.

\paragraph*{Related work and main results}
%We connect our approach to the approaches of \cite{asarin01perturbed} and  \cite{gao2012delta}. But observe that the logics considered in the latter are not sufficiently expressive to see our results or \cite{asarin01perturbed}, or ours, as a special case of \cite{gao2012delta}. 
%
%We believe that discussing when verification is decidable, and what level of robustness might be tolerated, is indeed  richer than these formalizations. 
%We are convinced that exploring relations between complexity of behaviours of a dynamical system (not necessarily related to verification or control questions, and of the so-called hybrid systems).

the current article lies its foundations on some extensions of \cite{asarin01perturbed}. We reformulate some of their statements, and extend some of their results. In particular, we allow more general discrete time and continuous time dynamical systems, while \cite{asarin01perturbed} was restricting to Piecewise Affine Maps (PAM), and hybrid systems, both of them assumed to work on a bounded domain. %We consider here general domains and dynamics. 
Our statements are established in a wide generality by considering only computability of the domain or the dynamics. We know about generalizations that have also been obtained in \cite{BGHRobust10}, but focusing on dynamical systems as language recognizers, and mainly focusing somehow on generalizations of Theorem \ref{thbizarre} (\cite[Theorem 4]{asarin01perturbed}), while we discuss here reachability in a more natural and general way.

We also connect our approach to  $\delta$-decidability introduced in \cite{gao2012delta}. Observe that the logics considered in the latter are not sufficiently expressive to cover our results (or \cite{asarin01perturbed}).

Another clear difference is also that we talk about complexity issues (space, time), while all these references are only talking about computability issues. Furthermore, we consider not only space perturbations, but also time and length perturbations, and we prove that many of the results, even for computability, can also be extended to this framework.

%\olivierpasimportant{Possiblement parler de Braverman's work about pspace. Non indispensable peut etre}
%%\olivier{Etendre cela en disant plein d'autres trucs, dont liens avec \cite{gao2012delta}, Riccardo's work, Mes propres oeuvrees \cite{BGHRobust10} avec Hainry et Graca, etc.}

%
In Section \ref{sec:graphs}, we recall some known facts about the hardness of reachability (accessibility) problems in graphs, discussing in particular their encoding.

Before going to general discrete time dynamical systems, and later continuous and hybrid time dynamical systems, we consider a specific case of discrete time dynamical systems, namely Turing machines (TMs). We think this helps to understand the coming statements.  We basically introduce in  in Section \ref{sec:mt} the concept of (space) perturbed Turing machine, following \cite{asarin01perturbed}, and recall some of the results established in this paper. The original contributions in this section are the extension of this framework to complexity. We prove that considering a polynomial robustness corresponds precisely to polynomial space, i.e. $\PSPACE$ (Theorem \ref{mainpspaceone}). We consider here space perturbations, but we prove later (in Section  \ref{sec:otherperturb}) that a natural concept of time perturbation  could also be considered in order to lead to  a characterisation of polynomial time, i.e. $\PTIME$ (Theorem \ref{mainptimeone}), using similar ideas.

Section \ref{sec:simulation}, recalls how several undecidability results for dynamical systems are established using the embedding of a class of dynamical systems (e.g. TMs) into another. We believe this helps to understand the related sources of non-robustness.

Section \ref{sec:discretetime} is considering general discrete time dynamical systems (over $\R^{d}$). We use the concept of perturbation and robustness from \cite{asarin01perturbed}. But unlike the latter article, mainly restricting to Piecewise Affine Maps (PAMs), we consider general systems, discussing the importance of the domain, and of the preservation of rational numbers as in PAMs. A main technical result is then Theorem \ref{thcinq}: this is an extension of \cite[Theorem 5]{asarin01perturbed} to the general settings, established using similar ideas, but not exactly\footnote{As the proof there turns out to be ambiguous about  whether ``can make a transition'' means in one step vs in one or many steps, and we are not sure to fully follow it in its exact form.}. We also extend the framework to the case of systems that would not preserve rationals. This leads to an extension of main statement of \cite{asarin01perturbed}: robustness implies decidability of reachability  (Corollary \ref{coro6} and Corollary \ref{coro6p}).  Then we prove that this theory can actually be related to the approach of $\delta$-decidability proposed in \cite{gao2012delta}: ($\epsilon$-)robust means true or ($\epsilon$-)far from being true (Theorem \ref{prop28}). This is an also important original contribution relating the two approaches, giving arguments in the spirit of \cite{gao2012delta}. %: Undecidability does not hold for $\delta$-robust statements. 
Our key point is then to be able to extend all at the complexity level. One key result is that polynomially robust implies reachability in  $\PSPACE$ (Theorem \ref{thmaindeux}). This is even a characterization of $\PSPACE$ (Theorem \ref{recipir}).
Concretely, Section \ref{sec:discretetime} is split into two subsections. In the first part, we only need to talk about rational numbers, and classical computability. In the second part, as functions may take non-rational values, we use the framework of Computable Analysis (CA). The first part helps for the intuition of the second. But actually, as a computable function in CA is continuous, the first is not a consequence of the second, even if several methods and reasoning are shared.

Section \ref{Sec:analysecalculable} provides a nice and original view. Using various results established in the community of CA, we prove that robustness can also be seen as having a reachability relation that can be drawn. We use there the fact that the latter relates to computability for subsets of $\R^{p}$. 

In Section \ref{sec:continuoustime}, we state that similar statements hold also for continuous time, and even hybrid dynamical systems. 
%Once again, the statements there can be seen as extending \cite{asarin01perturbed} to  general dynamics, relating the approach also to the concept of $\delta$-decidability of \cite{gao2012delta}, but mainly dealing not only with computability but also with complexity. 

Section \ref{sec:otherperturb} is proving that other perturbations could be considered and would lead to talk about time complexity $\PTIME$ instead of $\PSPACE$. Section \ref{sec:jebrode} discusses how the current work relates to the known characterizations of complexity classes with continuous time dynamical systems.  
Notice that some characterizations of $\PSPACE$ have been obtained recently in \cite{TheseRiccardo,BGDPRiccardo2022} (discussed in this section), but at the price of technical conditions, that we believe not to provide a fully satisfying answer to previous questions.

%We solve this problem by proving that basically space  corresponds to precision, or if one prefers, in some rather elegant manner, this corresponds to having some abstraction whose logarithm of the size remains polynomial. 

%\olivier{Lister les contributions précises du papier. Fait jusqu'a section VI, mais pas fini.}

\paragraph*{Preliminaries}
a * means that a (possibly more extended) proof can be found in appendix.  
%We write $\dyadic$ for the set of \motnouv{dyadics}: Rational numbers of the form $\frac{p}{2^{q+1}}$ for some $p \in \Z$, $q \in \N$. %We say that a vector $\vx$ is rational (respectively dyadic) if all its components are. 
We write $\distance{\cdot}{\cdot}$ for the norm-sup distance. A (open) (respectively close) \motnouv{rational ball} is a subset of the form $B(\vx,\delta)=\{\vy: \distance{\vx}{\vy} < \delta\}$ (resp. $\cB(\vx,\delta)=\{\vy: \distance{\vx}{\vy} \le \delta\}$) for some rational $\vx$ and $\delta$. We could in principle use the Euclidean distance, but this distance has the advantage that its balls correspond directly to rounding at some precision. A set of the form $\prod_{i=1}^{d} [a_{i},b_{i}]$, for rational $(a_{i})$, $(b_{i})$, will be called a \motnouv{rational closed box}.  An open rational box is obtained by considering open intervals in previous definition. 
The least closed set containing $X$ is denoted by $\closure{X}$. We write $\length(\cdot)$ for the function that measures the binary size of its argument.

\section{About reachability in graphs}
\label{sec:graphs}

%\olivierpasimportant{attention, plein de classes non définies.}

Consider a directed graph $G=(V,\rightarrow)$, with $\rightarrow \subseteq V^{2}$. It is said \motnouv{deterministic} if any node has at most one outgoing edge.

\begin{lemma}[Reachability for graphs] \label{lemmapath}
Consider the following decision problem $\PATH(G,u,v)$: 
given a directed graph $G=(V,\rightarrow)$, and some vertices $u,v \in V$, determine whether there is some path between $u$ and $v$ in $G$, denoted by $u \stackrel{*}{\rightarrow} v$.

Then $\PATH(G,u,v) \in \NLOGSPACE$. 
\end{lemma}

\begin{proof}
Basically, a non-deterministic TM can guess the intermediate nodes of the path: see e.g. \cite[Example 8.19]{Sip97}. 
\end{proof}

\begin{lemmad}[Immerman–Szelepcsényi's theorem \cite{immerman1988nondeterministic,Szelepcsenyi}] \label{lemma2}
$\NLOGSPACE=\coNLOGSPACE$.
\end{lemmad}

\begin{proof}
See e.g. \cite[Theorem 8.22]{Sip97}.
\end{proof}

Actually, as we will see, we mainly focus  its complement:

\begin{corollary}\label{coropath}
Consider the following decision problem 
\noindent $\LOOP(G,u,v)$: given a directed graph $G=(V,\rightarrow)$, and some vertices $u,v \in V$, determine whether there is  no path between $u$ and $v$ in $G$.

Then $\LOOP(G,u,v) \in \NLOGSPACE$. 
\end{corollary}

\begin{theoremd}[Savitch's theorem] \label{savitch} For any function $f: \N \to \N$ with $f(n) \ge \log n$, we have 
$\NSPACE(f(n)) \subseteq \sSPACE(f^{2}(n))$.
\end{theoremd}

\begin{proof}
See e.g. \cite[Theorem 8.5]{Sip97}.
\end{proof}

\begin{corollary}
$\PATH(G,u,v) \in \sSPACE(log^{2}(n))$ and 

\noindent $\LOOP(G,u,v) \in \sSPACE(log^{2}(n))$.
\end{corollary}

%\olivierpourmanon{En fait, NOPATH pourrait s'appeler LOOP, si ca aide à (ton/l') intuition. Changer? Mais en fait, à cause de la remarque d'apres c'est peut misleading? Avis?}
%\manon{NOPATH c'est bien}

\begin{remark}
Notice that detecting whether there is no path between $u$ and $v$ is basically equivalent to determine whether all paths starting from $u$ ``loop'', i.e. remain disjoint from $v$. %There is a clear  algorithm based on a depth or width search exploration of the graph, but the above statement is at the end established 
The above statement is established using a more subtle method that a simple depth of width search of the graph, using the trick of the proof of Savitch's theorem, i.e. a recursive procedure (expressing reachability in less than $2^{t}$ steps, called $\cp{CANYIELD(C_{1},C_{2},t})$ in \cite{Sip97}) guaranteeing the above space complexity. 
\end{remark}

Our purpose is to talk about various discrete time dynamical systems (and latter even continuous time dynamical systems). 

\begin{definition}[Discrete Time Dynamical System]  A (general) \motnouv{discrete-time dynamical system} $\mP$ is given by a set $X$, called \motnouv{domain}, and some  (possibly partial) function  $ f$ from $X$ to $X$.
\end{definition}

A trajectory of $\mP$ is a sequence $(x_t)$ evolving according to $f$, i.e. such that ${x}_{t+1}= f\left({x}_t\right)$ for all $t$. We say that $x^{*}$ (or a set $X^{*}$) is reachable from $x$ if there is a trajectory with $x_{0}=x$ and $x_{t}=x^{*}$ (respectively $x_{t} \in X^{*}$) for some $t$.

 In other words, any discrete time dynamical system $\mP$ can be seen as a particular (deterministic) directed graph but where $V$ is not necessarily finite. 
 This graph corresponds to $V=X$, and $\rightarrow$ corresponds to the graph of function $f$.
If it remains finite, we can generalize some of the previous statements, but working over representations in order to make things feasible.

\begin{corollaryd}[Reachability for finite graphs] \label{corosuccint}
Let $s(n) \ge \log(n)$. 
Assume the vertices of $G$ can be encoded in binary using words of length $s(n)$. Assume the relation $\rightarrow$ is decidable using a space polynomial in $s(n)$  with this encoding. 

Then, given the encoding of $u \in V$ and of $v \in V$, we can decide whether there is a (respectively: no) path from $u$ to $v$,  in a space polynomial in $s(n)$.
\end{corollaryd}

\begin{proof}
Use similar arguments and algorithms (in particular the trick of Savitch's theorem) as in previous corollary, but working over representations of vertices. 
\end{proof}

We  will still write (abusively) $\PATH$ and $\LOOP$  for these problems.  Notice that assuming that the vertices of $G$ can be encoded in binary using words of length $s(n)$ requires the graph $G$ to be finite, with less than $2^{s(n)}$ vertices. 

\begin{remark}
If write $\logsize(G)$ for the log of the number of vertices of a finite graph, we see that this provides some relevant complexity measure of the hardness of the space complexity of these problems (with above  assumptions on $G$).
\end{remark}

\section{Turing machines}
\label{sec:mt}

%\olivierpourmanoninfo{ajouté états de rejets}
%\olivierpourmanoninfo{En fait, je réalise qu'ils n'ont pas d'état de rejet. Du coup, ca veut dire quoi ``halts'' dont ils parlent? C'est assez-foireux leur truc.}
Let us recall the definition of a (bi-infinite tape) Turing machine (TM):  let $\Sigma$ be a finite alphabet, and let $B \not\in \Sigma$ be the blank symbol. A TM over $\Sigma$ is a tuple $\left(Q, q_{\text {init }}, F, R, \Gamma\right)$ where $Q$ is a finite set of control states, $q_{0} \in Q$ is the initial control state, $F \subseteq Q$ (respectively $R \subseteq Q$) is a set of accepting (respectively rejecting) states,  with $F \cap R =\emptyset$, and $\Gamma$ is a set of transitions of the form $(q, a) \rightarrow\left(q^{\prime}, b, \delta\right)$ where $q, q^{\prime} \in Q$, $a, b \in \Sigma \cup\{B\}$, and $\delta \in\{-1,0,1\}$. When the machine has accepted or rejected, decision is unchanged: when $q \in F$, then $q' \in F$, and when $q \in R$ then $q' \in R$.

A configuration $C$ of the machine is given by the current control state $q$, and the current content of the bi-infinite tape: ${\cdots  a_{-2} a_{-1}}{ a_{0} a_1 a_2 \cdots}$,  where the $a_i$'s are symbols in $\Sigma \cup\{B\}$: this means that the head of the machine is in front of symbol $a_0$. 
We write $\CONFIGURATIONSMT$ for the set of the configurations of a TM, and write such a configuration as the triple $\CONFIGMT{q}{\cdots  a_{-2} a_{-1}}{ a_{0} a_1 a_2 \cdots}$.
%: Second and third arguments are words over this alphabet. For helping readability, we write the word corresponding to second argument from right to left in this representation.  
%
Given a transition $(q, a) \rightarrow\left(q^{\prime}, b, \delta\right)$ in $\Gamma$, if the control state is $q$ and the symbol pointed by the head of the machine is equal to $a$, then the machine can change its configuration $C$ to the configuration $C'$ in the following manner: the control state is now $q'$, the symbol pointed by the head is replaced by $b$ and then the head is moved to the left or to the right, or it stays at the same position according to whether $\delta$ is $-1,1$, or 0, respectively. We write $C \vdash C'$ when this holds, i.e. $C'$ is the one-step next configuration of the configuration $C$. Then $(\CONFIGURATIONSMT,\vdash)$ corresponds to a particular dynamical system. 

Word $w=a_1 \cdots a_n \in \Sigma^*$ is accepted by $\M$ if, starting from the initial configuration $C_{0}=C_{0}[w]=\CONFIGMT{q_{0}}{\cdots B B B}{ a_1 a_{2} \cdots a_n B B B \cdots}$ the machine eventually stops in an accepting control state: that is, if we write $\mF$ for the configurations where $q \in F$, iff $C_{0}  \stackrel{*}{\vdash} C^{*}$ for some $C^{*} \in \mF$. Let $L(\M)$ denote the set of such words, i.e., the computably enumerable (c.e) language semi-recognized by $\M$.
We say that $w$ is rejected by $\M$ if, starting from the configuration $C_{0}$  the machine $\M$ eventually stops in a rejecting state. $\M$ is said to always halt if for all $w$, either $w$ is accepted or $w$ is rejected.

\olivierplan{section{Space perturbation (for Turing machines)}}

%\olivierpasimportant{pour des questions de rédaction, je sépare ce qu'ils ont fait de ce qu'on fait.}

The article \cite{asarin01perturbed} introduces the concept of space perturbed Turing machine: the idea is, given $n>0$, that the $n$-perturbed version of the machine $\M$ is unable to remain correct at distance more than $n$ from the head of the machine.
Namely,  the $n$-perturbed version $\M_{n}$  of the machine is defined exactly as $\M$ except that before any transition all the symbols at the distance $n$ or more from the head of the machine can be altered: %(i.e., replaced by other symbols) arbitrarily: 
given a configuration
$\CONFIGMT{q}{\cdots a_{-n-1} a_{-n} a_{-n+1} \cdots a_{-1}}{ a_{0 }a_1 \cdots a_{n-1} a_n a_{n+1} \cdots}$, 
$\M_{n}$ may replace any symbol to the left of $a_{-n}$ (starting from $\left.a_{-n-1}\right)$ and to the right of $a_n$ (starting from $a_{n+1}$) by any other symbols in $\Sigma \cup\{B\}$ before executing a transition of $\M$ at head position $a_{0}$. Hence $\M_{n}$ is nondeterministic. % (see figure 1 (b)).
\olivier{Des figures? Là, ou ailleurs?}

A word $w$ is accepted by the $n$-perturbed version of $\M$ iff there exists a run of this machine which stops in an accepting state. Let $\PERTURBEDS{\M}{n}$  be the $n$-perturbed language of $\M$, i.e., the set of words in $\Sigma^*$ that are accepted by the  $n$-perturbed version of $\M$.
From definitions, if a word is accepted by $\M$, then it can also be recognized by all the $n$-perturbed versions of $\M$, for every $n>0$: perturbed machines have more behaviours. Moreover, if the $(n+1)$-perturbed version accepts a word $w$, the $n$-perturbed version will also accept it since $n$-perturbed machines have more behaviours than $(n+1)$-perturbed machines.

%\begin{lemma}
%[{\cite[Lemma $1$]{asarin01perturbed}}]
%$$L(\M) \subseteq \cdots \subseteq \PERTURBEDS{\M}{2}  \subseteq \PERTURBEDS{\M}{1}.$$
%\end{lemma}

Let $L_\omega(\M)=\bigcap_n \PERTURBEDS{\M}{n}$: this 
consists of all the words that can be accepted by $\M$ when subject to arbitrarily ``small" perturbations. 
%
%
%%\olivier{variante}
%%
%%$$
%%L^{\omega}(\M)=\bigcap_n \PERTURBEDT{\M}{n}
%%$$
%
%%\olivier{Some text and results stolen from \cite{asarin01perturbed}}
%%
From definitions: 

\begin{lemma}
[{\cite[Lemma $2$]{asarin01perturbed}}] \ \\
\centerline{
$L(\M) \subseteq  L_\omega(\M)  \subseteq \cdots \subseteq\PERTURBEDS{\M}{2}  \subseteq \PERTURBEDS{\M}{1}.$}
\end{lemma}

The  $\omega$-perturbed language of a TM is the complement of a computably enumerable language:

\begin{theorem}
[{Perturbed reachability is co-c.e. \cite[Theorem 3]{asarin01perturbed}}]\label{LspaceCoRec}
 $L_\omega(\M)$ is in the class $\Pi_1^0$.
\end{theorem}

\olivierplan{subsection{Mix de leur preuve et de la preuve de Manon}}

\begin{proof}
Given a bi-infinite configuration $C$ of $M$ of the form $\CONFIGMT{q}{\cdots a_{-n-1} a_{-n} \dots  a_{-1}}{a_{0} a_1  \dots a_n a_{n+1} \cdots}$, 
we define 
$
\varphi_{n}(C) = %\left.c\right|_n=
\CONFIGMT{q}{a_{-n}  \cdots a_{-1}}{a_0  a_1 \cdots  a_n}
$ made of  words 
of length $n$ and $n+1$.

For every $n \in \N$, we associate to the $n$-perturbed version $\M_{n}$ of TM $\M$ some  graph $G_{n}=(V_{n},\rightarrow_{n})$: the vertices, denoted  $(\VERTEX_{i})_{i}$, of this graph correspond  to the $|Q| \times|\Sigma+1|^{2 n+1}$ possible values of $\varphi_{n}(C)$ for a configuration $C$ of $\M$.   There is an edge between  $\VERTEX_i$ and $\VERTEX_j$ in $G_{n}$  iff  $\M_{n}$ can go from configuration $C$ to configuration $C'$ in one step, with $\varphi_{n}(C)=\VERTEX_i$ and $\varphi_{n}(C')= \VERTEX_j$. 

Determining whether  $\VERTEX_{i} \rightarrow \VERTEX_{j}$  holds is easy (and in particular polynomial space computable) by  considering that, when the head is moved to the left (resp. to the right) of $\VERTEX_{i}$ a symbol in $\Sigma \cup\{B\}$ is nondeterministically chosen and appended to the left (resp. right) of the configuration and the right-most (resp. left-most) one is lost (it belongs now to the perturbed area of the configuration and hence it can be replaced by any other symbol).

Let $\mF_{n} =\varphi_{n}(\mF)$ correspond to the accepting control states.
By construction, the  $n$-perturbed version $\M_{n}$ of $\M$ has an accepting run starting from a configuration $C$, iff $\mF_{n}$ is reachable from $\varphi(C)$, that is to say $\PATH (G_{n},\varphi(C),\mF_{n})$. By Corollary \ref{corosuccint}, this is decidable in a space polynomial  in $n$.

Let $Basis_{n}$  be the finite set of sequences $s_{n}  \in \Sigma^{n+1}$, such that $\mF_{n}$ is reachable from $C_{0}[s_{n}]$. Let $Short_{n}$ be the finite set of sequences $s_{k}  \in \Sigma^{k}$ with $k<n$, such that $\mF_{n}$ is reachable from $C_{0}[s_{k}B^{n-k}]$.  Then $L_n(\M)=\operatorname{Short} \cup \operatorname{Basis} \Sigma^*$.
Consequently,   $L_{n}(\M)$ is decidable in space polynomial  in $n$, and hence its complement also is. Thus, $L_\omega(\M)$ is c.e., as it is a (uniform in $n$) union of decidable sets. 
%
%\olivierpasimportant{Etre encore plus précis.}
%Consequently, determining whether some input is not in $L_\omega(\M)$ is equivalent to determining whether there is no path between two sets of notes in this finite graphs, i.e. a loop is formed, which means, once 
%the machine is visiting again a vertex, then the corresponding graph 
%will visit only vertices already visited, but none of them intersects a terminating 
%one. 
%Since there is a finite number of vertices, the loop can be detected by a Turing 
%machine. Thus, $L_\omega(\M) \in \Pi_1^0$, at it basically requires to find an $n$ such that this decidable property holds.
\end{proof}

Since a set that is c.e. and co-c.e. is decidable,  robust languages (i.e. $L_\omega(\M)=L(\M)$) are necessarily decidable.

\begin{corollary}
[{Robust $\Rightarrow$ decidable \cite[Corollary 3)]{asarin01perturbed}}]
 If $L_\omega(\M)=L(\M)$ then $L(\M)$ is decidable.
\end{corollary}

The converse holds if another requirement on $\M$ is added.
%
%\olivierpourmanoninfo{New texte: je pense que ca aide à comprendre, et c'est plus propre que leur truc.}
Indeed, since $L(\M) \subseteq L_\omega(\M)$, $L_\omega(\M)\neq L(\M)$ means that there exists some word $w$, rejected by $\M$, but accepted by  any $n$-perturbed version $\M_{n}$. This $w$ is not rejected by $\M$ in finite time, otherwise it would use finitely many cells of the tape, and with $n$ sufficiently big, $\M_{n}$ would still reject it. In other words, this $w$ must nor be accepted nor rejected by $\M$.

\begin{proposition}
[{Decidable $\Rightarrow$ robust \cite[Proposition 1]{asarin01perturbed}}]  \label{decidablerobuste}

Assume $\M$ always halts. Then $L(\M)$ is decidable and 
$L_\omega(\M)=L(\M)$.

\end{proposition}

In general, $\omega$-perturbed languages are not computable enumerable. Some of them are complete among co-r.e. languages: perturbed reachability is complete in $\Pi_1^0$ \cite[Theorem 4]{asarin01perturbed}.

\begin{theorem}
[{Perturbed reachability is complete in $\Pi_1^0$ \cite[Theorem 4]{asarin01perturbed}}] \label{thbizarre}
 For every TM $\M$, we can effectively construct another TM  $\M^{\prime}$ such that $L_\omega\left(\M^{\prime}\right)=\overline{L(\M)}.$
\end{theorem}

\olivierplan{section{New results: SPACE (For Turing machines)}}

Actually is possible to go to some complexity issues, and not only restrict to computability.

%\begin{proof}
%This was established explicitly in the above proof. 
%\end{proof}

%\olivierpasimportant{
%\begin{proofmanon}
%	Let $n\in \N$.
%	We are going to prove the equivlalent statement : $$L_{n}(\M) \in 
%	NSPACE(poly(n)),$$ since, by Savitch theorem: $$NSPACE(poly(n)) = 
%	DSPACE(poly^2(n))$$
%	
%	We use again the discretisation in finitely many vertices of the space of configurations 
%	described for Theorem \ref{LspaceCoRec}.
%	
%	We execute $\M$, the machine $\mathcal{M'}$ recognizing 
%	$L_n(\M)$ go from one state to the next by applying the corresponding 
%	transition in $\M$ and must only memorise the current configuration of 
%	$\M$ (and not the whole execution), which a polynomial in $n$. 
%	
%	The introduction of non-determinism comes from the fact that from one vertex, one 
%	transition in $\M$ can lead to several different vertices, so different 
%	states in $\mathcal{M'}$, but there is a finite numbers of vertices.
%	
%	So we have 
%	$L_{n}(\M) \in NSPACE(poly(n))$, and we conclude.	
%	
%%	\olivier{un peu preuve avec les mains, mais ok pour l'instant.}
%%	\manon{c'est un peu mieux}
%	%\olivier{En fait, on peut aussi raisonner sur leur automate $A_{\M}$, ca donne la meme conclusion. }
%	%\manon{la topologie c'est plus joli (et ça rime en plus)}
%
%\end{proofmanon}
%\o}

Indeed, when a language is robust, 
%for any word $w$, there must exist some $n$ (depending possibly on $w$) such that $w \in L$ and  $w \in L_{n}(w)$ have the same truth value: $n$ can be read as the associated tolerated level of perturbation for $w$.
%
 it makes sense to measure what level of perturbation $f$ can be tolerated:

\begin{definition}
Given some function $f: \N \to \N$, we write $\PERTURBEDSPACE{\M}{f}$ for the set of words accepted by $\M$ with space perturbation $f$: 
	{$\PERTURBEDSPACE{\M}{f} = \{w | w\in \PERTURBEDS{\M}{f(\length(w))}\}.$}
\end{definition}

The  proof above of Theorem \ref{LspaceCoRec} establishes explicitly: 

%But actually, this is possible to go to some complexity issues, and not restrict to computability only:

\begin{lemma}\label{lemmadirectionun}
$\PERTURBEDS{\M}{n} \in \sSPACE(poly(n))$.
\end{lemma}

We get a characterization of $\PSPACE$:

\begin{theorem}[Polynomial precision robust $\Leftrightarrow$ $\PSPACE$] \label{mainpspaceone}
$L \in \PSPACE$ iff for some $\M$ and some polynomial $p$, $L=L(\M)= \PERTURBEDSPACE{\M}{p}.$
\end{theorem}

%\olivier{En fait, nos machines n'ont pas de worktape, et il y a un truc potentiel sur la définiton de pspace du coup. En tous cas, parler de worktape dans une preuve pose un pb}

\begin{proofmanon}
	($\Rightarrow$) If $M$ always terminates and works in polynomial space, then there exists a polynomial $q(\cdot)$ that bounds the size of the used part of the tape of $M$. Considering a polynomial $p \ge q+2$, we have for $n\in N$ $\PERTURBEDS{\M}{p(\length(w))}  \subseteq L(M)$. We always have the other inclusion.
		
	($\Leftarrow$) We always have $\PERTURBEDS{\M}{p(n)} \in \PSPACE$ by previous lemma, and since $\PERTURBEDS{\M}{p(n)} = L$, then $L \in PSPACE$.
\end{proofmanon}

This considers space perturbations. Other types of perturbations are considered in Section \ref{sec:otherperturb}, leading to $\PTIME$ instead of $\PSPACE$. 
For now, we keep to space perturbations. %and the equivalent for general  dynamical systems.

%\olivierpasimportant{POUR OLIVIER: Il reste ces morceau de prevue: de quoi? ca contient quelquechose d'important?
%
%\begin{proofmanon}
%	We have $L(M) \subseteq L_\omega(M)$.
%	
%	By contradiction, we suppose that $ L(M) \subsetneq L_\omega(M) $. Thus, there exists $x \in  L_\omega(M) $ and $x\notin L(M)$, so $M$ stops after using a space $f(|x|)$, so for $n\geq f(|x|)$, $x\notin L_\omega(M)$. Which is a contradiction. 
%\end{proofmanon}
%
%
%\begin{theoremaprouver}
%A language $L$ is in $PSPACE$ iff for some polynomial $p(n)$, 
%$$L=L(\M)=L_{\omega}(\M) = L_{p(n)}(\M)$$
%for some always terminating Turing machine $\M$.
%\end{theoremaprouver}
%Mais n'est pas le truc d'avant en fait.
%}

\section{Embedding dynamical systems}
\label{sec:simulation}

%\olivierpourmanoninfo{New section. L'idée c'est de parler codage d'un système dans un autre. En lien avec le fameux Master's report de Manon}

%A TM  is a very specific discrete time dynamical system over some domains $\CONFIGURATIONSMT$. 

Discussing issues for TMs has  the advantage that related computability and complexity issues are well-known or easier to discuss.
Many authors have then embedded TMs in various classes of dynamical systems in order to get hardness results, i.e. 
state that the difficulty of the reachability problem for the latter is at least as hard as  for Turing machines.

Generally speaking, the trick is the following:  if we fix the alphabet  to  $\Sigma=\{0,1\}$, and  $Q=\{1,\dots,q\}$ for some integer $q$,  and if we forget about blanks,  we can always consider that $\CONFIGURATIONSMT \subseteq \mathcal{C}=\N \times \Sigma^{*} \times \Sigma^{*}$, i.e. that a configuration is given by some control state, and two finite words.

\newcommand\gammac{\Upsilon}
%\olivier{Pourquoi je peux pas mettre en gras un gamma. Serait mieux que cete notation, non, a coup de bar pour distinguer $\gamma$ de $\gammac$?}

Fix some encoding function of configurations into a vector of real (or integer) numbers:  $\gammac: \mathcal{C} \to \R^{d}$, with $d \in \N$.
For example, one can consider, $\gammac(q,w_{1},w_{2})= (q,\gamma(w_{1}),\gamma(w_{2}))$ with $\gamma: \Sigma^{*} \to \R$ taken as:
\newcommand\gammaentier{\gamma_{\N}}
\newcommand\gammaentierk{\gamma_{\N}^{k}}
\newcommand\gammacompact{\gamma_{[0,1]}}
\newcommand\gammacompactk{\gamma_{[0,1]}^{k}}
\newcommand\gammacompactpk{\gamma_{[0,1]}^{'k}}
\begin{itemize}
\item  the encoding $\gammaentier$  that maps the word $w=a_{1} \dots a_{n}$ to the integer whose binary expression is $w$,
\item or  the encoding $\gammacompact$  that maps $w$ to the real number of $[0,1]$ whose binary expansion is $w$,
\item or  more generally, the encoding $\gammacompactk$  or $\gammaentierk$, using base $k$  instead of base $2$ for some $k \ge 2$,
\item or  $\gammacompactpk$ that maps $w$ to $(\gammacompactk(w),\length(w))$. 
%\item or variations on these, other encodings.
\end{itemize}

Assume you have a function $\tu f: X \subseteq \R^{d} \to X$ such that for any configuration $C$, if we denote by $C'$ the one step next configuration, we have $f(\gammac(C))=\gammac(C')$: i.e. one step of the Turing machine corresponds to one step of the dynamical system $(X,f)$ with respect to the encoding $\gamma$. That is,  the following diagram commutes for one step:
		\begin{center}
			\tikzset{every node/.style={align=center}} 
			% To align all nodes as centered (Thank you, CarLaTeX, for showing me this)
			\begin{tikzpicture}[x=1cm,y=0.8cm]
			\node at (0,0) (a) {$C$};
			\node at (2,0) (b) {$C'$};
			\node at (0,-2) (aa) {$\gammac(C)$};
			\node at (2,-2) (bb) {$\gammac(C') $};

			\node at (1, 0.3) (ah) {$\vdash$};
			\node at (2.2, -1) (i) {$\gammac$};
			\node at (1, -1.7) (h) {$\tu f$};
			\node at (-0.3, -1) (i) {$\gammac$};
			\draw [->] (a)--(b);
			\draw [->] (b)--(bb);
			\draw [->] (aa)--(bb);
			\draw [->] (a)--(aa);
			\end{tikzpicture}
		\end{center}
		
Then it will commutes for any number of steps:
		\begin{center}
			\tikzset{every node/.style={align=center}} 
			% To align all nodes as centered (Thank you, CarLaTeX, for showing me this)
			\begin{tikzpicture}[x=1cm,y=0.8cm]
			\node at (0,0) (a) {$C$};
			\node at (0,-2) (aa) {$\gammac(C)$};

			\node at (2,0) (b) {$C'$};
			\node at (2,-2) (bb) {$\gammac(C') $};
			\node at (4,0) (c) {$C''$};
			\node at (4,-2) (cc) {$\gammac(C') $};
			\node at (6,0) (d) {$C'''$};
			\node at (6,-2) (dd) {$\gammac(C''') $};
				\node at (8,0) (e) {};
			\node at (8,-2) (ee) {};

			\node at (-0.3, -1) (i) {$\gammac$};
	
			\node at (1, 0.3) (ah) {$\vdash$};
			\node at (1, -1.7) (h) {$\tu f$};
			\node at (2.2, -1) (i) {$\gammac$};

 		         \node at (3, 0.3) (ah) {$\vdash$};
			\node at (3, -1.7) (h) {$\tu f$};
			\node at (4.2, -1) (i) {$\gammac$};

 		         \node at (5, 0.3) (ah) {$\vdash$};
			\node at (5, -1.7) (h) {$\tu f$};
			\node at (6.2, -1) (i) {$\gammac$};
			
 		         \node at (7, 0.3) (ah) {$\vdash$};
			\node at (7, -1.7) (h) {$\tu f$};
%			\node at (8.2, -1) (i) {$\gammac$};

			\draw [->] (a)--(aa);	
					
			\draw [->] (a)--(b);
			\draw [->] (b)--(bb);
			\draw [->] (aa)--(bb);
			
			\draw [->] (b)--(c);
			\draw [->] (c)--(cc);
			\draw [->] (bb)--(cc);
		
			\draw [->] (c)--(d);
			\draw [->] (d)--(dd);
			\draw [->] (cc)--(dd);
			
			\draw [->,dashed] (d)--(e);
			\draw [->,dashed] (dd)--(ee);		
%			\draw [->,dashed] (e)--(ee);		
		
			\end{tikzpicture}
		\end{center}

And hence, questions related to the existence of trajectories in the (dynamical system  associated to the) Turing machine will be mapped to corresponding questions about the existence of trajectories over the dynamical system $(X,f)$.

In particular, as reachability is undecidable (c.e. complete, and hence c.e. hard) for Turing machines,  this provides undecidability (c.e. hardness) of reachability for various classes of dynamical systems. As in most of the natural classes of dynamical systems, reachability is c.e. (just simulate the system to get a semi-algorithm), this leads to c.e. completeness.

Call such a situation a \motnouv{step-by-step} emulation.

%
%\olivier{sauf qu'il faut compléter, et pour l'instant c'est pas fait.}
%
%
%\olivier{Pas content sur l 'ordre des choses en l'état.
%
%Pour l'instant, envie d'écrire:
%This idea has been used with various encoding functions $\gamma$, and various classes of dynamical systems with such functions $\tu f$.
%
%\begin{enumerate}
%\item Using an encoding similar to $\gammacompact$, and $\tu f$ piecewise affine, the authors of \cite{Moo91,KCG94} have established that Piecewise Affine Maps can step by step simulate Turing machines:
%A a consequence:
%\begin{itemize}
%\item  For any PAM $\mP$ its reachability relation is computably enumerable. 
%\item Any computably enumerate set $S$ is reducible to the reachability relation of a PAM.
%\end{itemize}
%\item \dots
%\end{enumerate}
%
%mais alors je le redis plus loin.
%}
%

Such embedding strategies do provide undecidability results. But,  encodings such as $\gammacompact$ or $\gammacompactk$, whose image is compact,  map some intrinsically different configurations to points arbitrarily closed to each other (as a sequence over a compact must have some accumulation point).
%, and hence the considered $\tu f$ are ``non-robust''. 
An encodings like $\gammaentier$ do not have a compact image, but involves emulations with arbitrarily big integers, which is another issue. %, and a source of ``non-robustness''.

\section{Discrete Time Dynamical Systems}
\label{sec:discretetime}

This leads to discuss now robustness issues for general dynamical systems over in $\R^{d}$ for some $d\in\N$.  

\begin{definition}[Discrete Time Dynamical System]  A \motnouv{discrete-time dynamical system} $\mP$ is given by a
set $X \subseteq \mathbb{R}^d$, and some  (possibly partial) function  $\tu f$ from $X$ to $X$.
\end{definition}

%A trajectory of $\mP$ is a sequence $\mathbf{\tu x}_n$ evolving according to $\tu f$, i.e. such that $\mathbf{x}_{n+1}=\tu f\left(\mathbf{x}_n\right)$ for all $n$.

The dynamical system will be called  \motnouv{rational} when $\tu f$ preserves rational numbers, i.e. whenever $\tu f(\Q^{d}) \subseteq \Q^{d}$. We will say that a system  is $\Q$-computable, if it is rational and $\tu f: \Q^{d} \to \Q^{d}$ is computable. We say that the system is (respectively: locally)  \motnouv{Lipschitz} when the function is.
\olivierpasimportant{Peut etre pas mettre la def.:  (resp. For all $\tu \vx \in X$, there is some neighborhood $\mathcal{V}(\vx)$ of $\vx$ such that)  there exists some constant $L$ such that $\distance{\tu f(\vx)}{\tu f(\vy)} \le L \cdot \distance{\vx}{\vy}$ for all $\vx,\vy \in X$ (resp. $\vy,\vy \in X \cap \mathcal V(\vx)$). }

To each rational discrete time dynamical system $\mP$ is associated its reachability relation $\REACHP(\cdot, \cdot)$ on $\Q^d\times \Q^{d}$. Namely, for two rational points $\mathbf{x}$ and $\mathbf{y}$,  the relation $\REACHP(\mathbf{x}, \mathbf{y})$ holds iff there exists a trajectory of $\mP$ from $\mathbf{x}$ to $\vy$.

\subsection{The case of rational systems}
\label{sub:partone}

We first focus on the case of rational systems. Clearly the reachability relation of a $\Q$-computable system is computably enumerable: just simulate the dynamics with a TM. 
Actually, \cite{asarin01perturbed} considers only the special case of Piecewise affine maps, as representative of discrete time systems, which are particular $\Q$-computable Lipschitz systems.

\begin{definition}[PAM System]  A Piecewise affine map system (PAM) is a discrete-time dynamical system $\mP$ where $\tu f$ is a (possibly partial) function from $X$ to $X$ represented by a formula:
$
\tu f(\mathbf{x})=A_i \tu{x}+\mathbf{b}_i \text { for } \vx \in P_i, \quad i=1 \dots N
$
where $A_i$ are rational $d \times d$-matrices, $\mathbf{b}_i \in \Q^d$ and $P_i$ are convex rational polyhedral sets in $X$.
\end{definition}

%A trajectory of $\mP$ is a sequence $\mathbf{x}_n$ evolving according to $f$, i.e. such that $\mathbf{x}_{n+1}=f\left(\mathbf{x}_n\right)$ for all $n$.

In other words, a PAM system consists of partitioning the space into convex polyhedral sets (called \motnouv{regions}), and assigning an affine update rule $\mathbf{x}:=A_i \mathbf{x}+\mathbf{b}_i$ to all the points sharing the same region.

%To each PAM $\mP$ is associated its reachability relation $\REACHP(\cdot, \cdot)$ on $\Q^d\times \Q^{d}$. Namely, for two rational points $\mathbf{x}$ and $\mathbf{y}$ the relation $\REACHP(\mathbf{x}, \mathbf{y})$ holds iff there exists a trajectory of $\mP$ from $\vx$ to $\vy$.

\begin{remark}
All constants in the PAM definitions %and  $\mathbf{x}$ and $\mathbf{y}$ 
are assumed to be rational so that this remains a $\Q$-computable system. No form of continuity is assumed on function $\tu f$. %In particular so that the expressive power of PAM and the hardness of questions is not achieved using the introduction of some non-computable real numbers.
\end{remark}
%

%\olivier{Condenser ces deux trucs. $1$-reductible, soit expliquer, soit oublier.}

The following result  on the computational power of PAMs is known, and has been established using the technique of step-by-step emulation described in previous section (using  $\gammacompact$ and taking $\tu f$ as piecewise affine). 

\begin{theorem}[{Computational power of PAMs \cite{Moo91} \cite{KCG94}%, \cite[Theorem 2]{asarin01perturbed} 
}] \label{chosemachin}
\  
%
%Let $\M$ be a TM. We can effectively construct a $P A M ~ \mP$ and an encoding $e: \Sigma^* \rightarrow$ $ Q^d$ such that for any word $w$ the following equivalence holds: $w \in L(\M)$ iff $\REACHP(e(w), O)$, where $O$ denotes the origin in $\mathbb{R}^d$.
%%\end{theorem}
%%
%%
%%\begin{corollary}[{\cite[Corollary 1]{asarin01perturbed} (Computational power of PAM)}]
%
%As a consequence ({\cite[Corollary 1]{asarin01perturbed}}):
%
%For any PAM $\mP$ (as for any rational computable over the rational system)  its reachability relation is computably enumerable.  
Any c.e. language is reducible to the reachability relation of a PAM.
\end{theorem}

\begin{remark}
PAMs are introduced in \cite{asarin01perturbed}  only for the case where $X$ is necessarily some bounded polyhedral sets.  Actually, from considered $\gamma$,  the above result would also still hold when the regions $P_{i}$ are assumed to be rational boxes.
%
%\olivier{Sauf que si régions polyhedrals, comment est-ce possible que $X$ ne le soit pas. Bn c'est une remarque de rédaction philosophique.  Reste borné vs non-borné: \c ca, ok}
\end{remark}

\olivierplan{subsection{Perturbed PAMs (PPAMs)}}

\olivierplan{Some text and results stolen from \cite{asarin01perturbed}.}

This proves c.e.-completeness ($\Sigma_1^0$-completeness), and hence undecidability  of reachability for  $\Q$-computable systems. Let  discuss whether this  still holds for ``robust systems''. % but not restricting only to PAMs like  \cite{asarin01perturbed}.

We can apply the paradigm of small perturbations: consider a discrete time dynamical system $\mP$ with function $\tu f$. For any $\varepsilon>0$ we consider the $\varepsilon$-perturbed system $\mP_{\varepsilon}$. Its trajectories are defined as sequences $\mathbf{x}_t$ satisfying the inequality $\distance{\mathbf{x}_{t+1}}{\tu f\left(\mathbf{x}_t\right)}<\varepsilon$ for all $t$. This non-deterministic system can be considered as $\mP$ submitted to a small noise with magnitude $\varepsilon$. For convenience, we write $\tu y \in \tu f_{\epsilon}(\vx)$ as a synonym for $\distance{\tu f(\vx)}{\tu y} < \epsilon$. 
We denote reachability in the system $\mP_{\varepsilon}$ by $\REACHPepsilon(\cdot, \cdot)$. 

All trajectories of a non-perturbed system $\mP$ are also trajectories of the $\varepsilon$-perturbed system $\mP_{\varepsilon}$. If $\varepsilon_1<\varepsilon_2$ then any trajectory of the $\varepsilon_1$-perturbed system is also a trajectory of the $\varepsilon_2$-perturbed PAM.
Like for TMs we can pass to a limit for $\varepsilon \rightarrow 0$. Namely $\REACHPomega(\mathbf{x}, \mathbf{y})$ iff $\forall \varepsilon>0$  $\REACHPepsilon(\mathbf{x}, \mathbf{y})$: this relation encodes  reachability with arbitrarily small perturbing noise. 
%
%From definitions: %, for rational systems:
\begin{lemma}[{\cite[Lemma 3]{asarin01perturbed}}] For any $0<\varepsilon_2<\varepsilon_1$ and any  $\mathbf{x}$ and $\mathbf{y}$ the following implications hold: $\REACHP(\mathbf{x}, \mathbf{y})$ $ \Rightarrow \REACHPomega(\mathbf{x}, \mathbf{y}) \Rightarrow \REACHP_{\varepsilon_2}(\mathbf{x}, \mathbf{y}) \Rightarrow \REACHP_{\varepsilon_1}(\mathbf{x}, \mathbf{y})$.
\end{lemma}

\olivierplan{subsection{Results on PPAMs}}

%\olivierpasimportant{Some text and results stolen from \cite{asarin01perturbed}.}

We  prove  the perturbed reachability relation of Lipschitz $\Q$-computable system  is co-c.e, extending \cite[Theorem 5]{asarin01perturbed}. %(a PAM is $L$ Lipschitz: take $L=\max_{i} \| A_{i} \|$).

\begin{theorem} 
[{Perturbed reachability is co-c.e.}]  \label{hypothesesquivontbien}\label{thcinq} \label{Rcorec}
Consider a locally Lipschitz $\Q$-computable system whose domain $X$ is a closed rational box.
%computable compact.
	
Then the relation $\REACHPomega(\mathbf{x}, \mathbf{y}) \subseteq \Q^d\times \Q^{d}$ is in the class $\Pi_1^0$.
\end{theorem}

\begin{proofmanon}

As $\tu f$ is locally Lipschitz, and $X$ is compact, we know that $\tu f$ is Lipschitz: there exists some $L>0$ so that $\distance{\tu f(\vx)}{\tu f(\vy)} \le L \cdot \distance{\vx}{\vy}$. 

	For every $\delta=2^{-m}$, $m \in \N$, we associate some  graph $G_{m}=(V_{\delta},\rightarrow_{\delta})$: its vertices, denoted  by $(\VERTEX_{i})_{i}$, correspond  to some finite covering of compact $X$ by rational open balls $\VERTEX_{i}=B(\vx_{i},\delta_{i})$  of  radius $\delta_{i} <\delta$. 
	%
	%We write %(abusively)	$\varphi_{\delta}(\vx)= \VERTEX_{i}$ whenever $\vx$ is covered by $\VERTEX_{i}$, and 
	%$\varphi_{\delta}(\vx)$ as the (union of the finitely many) vertex (vertices) covering $\vx$.
%
	There is an edge from $\VERTEX_i$ to $\VERTEX_j$ in this graph, that is to say $\VERTEX_{i} \rightarrow_{\delta} \VERTEX_{j}$,  iff 
	$B(\tu f(\vx_{i}),(L+1)\delta) \cap \VERTEX_{j} \neq \emptyset$.
%% EQUIVALENT:  $\distance{\tu f(\vx_{i})}{\vy}<(L+1)\delta$ for some $\vy \in \VERTEX_{j}$, that is to say i
%
With our hypothesis on the domain, such a graph can be effectively obtained from $m$, considering a suitable discretization of the rational box $X$.  %: determining whether two rational balls intersect is easy.	 

	 \olivierplan{Nouvelle formulation et amélioration}
	 
%	 \begin{itemize}
%	 \item 
\noindent Claim 1: 
	 assume $\REACHPepsilon(\vx,\vy)$ with $\vx \in \VERTEX_{i}$ for $\epsilon=2^{-n}$. Then $\VERTEX_{i} {\rightarrow_{\epsilon}} \VERTEX_{j}$ for all  $\VERTEX_{j}$ with $\vy \in \VERTEX_{j}$.
	 
	 This basically holds as the graph for $\delta=\epsilon$ is made to always  have more trajectories/behaviours than $\REACHPepsilon$. 
	\begin{proof} If $\vy \in \tu f_{\epsilon}(\vx)$, then $\distance{\tu f(\vx_{i})}{\vy} \le \distance{\tu f(\vx_{i})}{\tu f(\vx)} + \distance{\tu f(\vx)}{\vy} < L \distance{\vx_{i}}{\vx} + \epsilon \le L \epsilon + \epsilon = (L+1) \epsilon$, and hence there is an edge from $\VERTEX_{i} {\rightarrow_{\epsilon}} \VERTEX_{j}$ to any $\VERTEX_{j}$ containing $\vy$ by definition of the graph.
	\end{proof}
	
\noindent  Claim 2: for any $\epsilon=2^{-n}$, there is  some $\delta=2^{-m}$ so that if we have $\VERTEX_{i} \stackrel{*}{\rightarrow_{\delta}} \VERTEX_{j}$ then 
	$\REACHPepsilon(\vx,\vy)$ whenever $\vx \in \VERTEX_{i}$, $\vy \in \VERTEX_{j}$.
	
		 	That is, Claim 2 says that $\neg \REACHPepsilon(\vx,\vy)$ implies $\neg (\VERTEX_{i} \stackrel{*}{\rightarrow_{\delta}} \VERTEX_{j})$ whenever $\vx \in \VERTEX_{i}$, $\vy \in \VERTEX_{j}$, for the corresponding $\delta$.

	\begin{proof} Consider $\delta=2^{-m}$ with $\delta <\epsilon/(2L+2)$: assume  $\VERTEX_{i=i_{0}} {\rightarrow_{\delta}} \VERTEX_{i_{1}} \dots {\rightarrow_{\delta}}  \VERTEX_{i_{t}=j}$ with $\vx \in \VERTEX_{i}$, $\vy \in \VERTEX_{j}$. 
	
	Assume by contradiction that $\neg  \REACHPepsilon(\vx,\vy)$, and let  $\ell$ be the least index such that 
	%$\REACHPepsilon(\vx,\overline x)$ for all $\overline x \in  \VERTEX_{i_{\ell}}$ but but 
	$\neg \REACHPepsilon(\vx,\overline \vz)$ for some $\overline \vz \in \VERTEX_{i_{\ell+1}}$.
	
	As 	 $\VERTEX_{i_{\ell}} {\rightarrow_{\delta}} \VERTEX_{i_{\ell+1}}$ there is some $\overline {\tu y} \in  \VERTEX_{i_{\ell+1}}$ with $\distance{\tu f(\tu x_{i_{\ell}})}{\overline{\vy}}<(L+1) \delta$.   Take $\overline{\tu z} \in \VERTEX_{i_{\ell+1}}$. 
	
	If $\ell=0$, then $\distance{\tu f(\vx)}{\overline{\tu z}} \le \distance{\tu f(\vx)}{\tu f(\vx_{i_{\ell}})}+\distance{\tu f(\vx_{i_{\ell}})}{\overline{\vy}} + \distance{\overline{\vy}}{\overline{\vz}} < L \delta + (L+1) \delta + \delta = (2 L + 2) \delta < \epsilon$, and hence
	$\REACHPepsilon(\vx,\overline{\vz})$: contradiction.
	
	If $\ell>0$, as $\ell$ is the least index with the above property, $\REACHPepsilon(\vx,\vx_{i_{\ell}})$. But then 
	$\distance{\tu f(\vx_{i_{\ell}})}{\overline{\tu z}} \le \distance{\tu f(\vx_{i_{\ell}})}{\overline{\vy}} + \distance{\overline{\vy}}{\overline{\vz}}
	< (L+1) \delta + \delta < (2L + 2) \delta < \epsilon$. 
%	
%	Then $\distance{\overline z}{f(\vx_{i_{\ell}}} \le 
%	\distance{\overline z}{\vx_{i_{\ell+1}}} + \distance {\vx_{i_{\ell+1}}}{} $
%	
%	$\distance{\vx_{i+1}}{f(\vx_{i})}  \le \distance{ \vx_{i+1}}{ \overline{\vy} } + \distance{\overline{\vy}}{f(\overline{\vx_{}})} + \distance{f(\overline{\vx})}{f(\vx_{i})} \le \delta + \delta + L \delta = (L+2) \delta = \epsilon$.  
%	
	And hence, $\REACHPepsilon(\vx_{i_{\ell}},\overline{\vz})$, and since we have $\REACHPepsilon(\vx,\vx_{i_{\ell}})$, we get $\REACHPepsilon(\vx,\overline \vz)$ and a contradiction.
\end{proof}

%	 \end{itemize}

		 \olivierplusimportant{ Pour aider:
 etape 1	 
	 $\REACHPomega(\vx,\vy)$ holds iff for all $\delta=2^{-m}$, we have $\VERTEX_{i} \stackrel{*}{\rightarrow_{\delta}} \VERTEX_{j}$, for all  $\VERTEX_{i}$, $\VERTEX_{j}$ with  $\vx \in \VERTEX_{i}$, $\vy \in \VERTEX_{j}$.

 Etape 1'
	 	 If one prefers, 
}
		 
	 From the two above items, 
 $\neg \REACHPomega(\vx,\vy)$ holds iff for some $\delta=2^{-m}$, $\neg (\VERTEX_{i }\stackrel{*}{\rightarrow_{\delta}} \VERTEX_{j})$ for some $\VERTEX_{i}$, $\VERTEX_{j}$ with $\vx \in \VERTEX_{i}$, $\vy \in \VERTEX_{j}$.
	 This holds  iff for some integer $m$, $\LOOP(G_{m},\VERTEX_{i},\VERTEX_{j})$ for some  $\VERTEX_{i}$, $\VERTEX_{j}$ with $\vx \in \VERTEX_{i}$, $\vy \in \VERTEX_{j}$.

		The latter property is c.e., as it corresponds to a  union of decidable sets (uniform in $m$), as 		$\LOOP(G_{m},\VERTEX_{i},\VERTEX_{j})$ is a decidable property over finite graph $G_{m}$.
\end{proofmanon}

Notice  this would work even only assuming the domain to be a computable compact: we will  recall later what it is (using computable analysis), but let's say for now that the above proof only requires that given  $\delta=2^{-m}$, there is an effective way to determine an effective cover of it using finitely many rational balls of radius less than $\delta$. %This holds for computable compacts.

\begin{corollary}
[{Robust $\Rightarrow$ decidable \cite[Corollary 5]{asarin01perturbed}}]  \label{coro6}
Assume the hypotheses of Theorem \ref{hypothesesquivontbien}.

If $\REACHPomega=\REACHP$ then $\REACHP$ is decidable.
\end{corollary}

\begin{proofmanon}
	$\REACHP$ is c.e. and we know from  Theorem \ref{thcinq} that $\REACHPomega$ is co-c.e.. If they are equal, then they are decidable, as a c.e.  and co-c.e. set  is decidable.
\end{proofmanon}

%\olivier{Verifier attentivement ce délire suivant: est-ce bien exactement ca.}
\begin{remark}
Notice that a similar statement holds even if $X$ is not compact: from the proof, it is sufficient that there exists some family of graphs $\mG=(G_{m})$ with $G_{m}=(V_{m},\rightarrow_{m})$  to get a similar reasoning with the following properties:  
\begin{enumerate} \label{def:abstraction}
\item $\REACHPepsilon(\vx,\vy)$ with $\vx \in \VERTEX_{i}$, $\epsilon=2^{-n}$, implies $\VERTEX_{i} {\rightarrow_{n}} \VERTEX_{j}$ for  all  $\VERTEX_{j}$ containing $\vy$.

\olivierplusimportant{Etait: $\REACHPepsilon(\vx,\vy)$ implies $\PATH(G_{\epsilon}, \varphi_{\epsilon}(\vx),
	 \varphi_{\epsilon}(\vy))$}
\item For any $\epsilon=2^{-n}$, there is  some $m$ such that if we have $\VERTEX_{i} \stackrel{*}{\rightarrow_{m}} \VERTEX_{j}$ then 
	$\REACHPepsilon(\vx,\vy)$ whenever $\vx \in \VERTEX_{i}$, $\vy \in \VERTEX_{j}$.
\olivierplusimportant{Etait $\varphi_{\delta}(\vx) \rightarrow \varphi_{\delta}(\vy)$ in $G_{\delta}$ implies $\REACHP_{\epsilon(\delta)}(\vx,\vy)$, where  $\epsilon(\delta)$ is some function that goes to $0$ when $\delta$ goes to $0$.}

\item  For all $m$, $G_{m}$ is a  finite computable graph: determining whether $\VERTEX_{i} \rightarrow_{m} \VERTEX_{j}$ in $G_{m}$ can be effectively determined given integers $m$, $i$, and $j$.
\end{enumerate}

When these three properties hold, we say that $\mG$ is  \emph{a computable abstraction} of the discrete time dynamical system.
\end{remark}

%\olivier{NEW: Des devoirs pour manon. J'ai l'impressionq ue c'est vrai tout ca. Basé sur tes premiers amours, en cherchant à dire ce que voulait dire ``terminer'' pour un système dynamique}

%\manon{Du coup c'est quoi mes devoir :)? }

There is a kind of converse property if some condition is added. Before stated this as Corollary \ref{coroconversee}, we relate robustness to the concept of $\delta$-decidability in \cite{gao2012delta}, and also the existence of some witness of non-reachability.

Given $\vx$, we write $\REACHP(\vx)$ for the set of the points $\vy$ reachable from $\vx$:
$\REACHP(\vx)=\{\vy | \REACHP(\vx,\vy)\}.$ This is  easily seen to also corresponds the smallest set
such that $\vx \in \REACHP(\vx)$ and $\tu f(\REACHP(\vx)) \subseteq \REACHP(\vx)$. 

%Consequently, there is some subset $R \subseteq X$ such that 
%\begin{itemize}
%\item $\vx \in R$,
%\item $f_{\epsilon}(\overline{R}) \subseteq R$,
%\item $y \not\in R$ 
%\end{itemize}
%if and only if $\neg \Reach(\vx,\vy)$: Indeed, when $\Reach(\vx,\vy)$, then any solution of the first two items includes
%

\newcommand \oR{\mathcal{R}^{*}} 
\newcommand \ooR{{R}^{*}} 

 %\olivierpourmanoninfo{Pour info. J'ai changé  par rapport à hier vendredi le 3ième item (pour arriver à conclure pour  Corollaire \ref{coro29}) $\distance{\vy}{\oR} > \epsilon$ en $y \not\in \oR$. Et d'autre part,  j'ecris plus $\overline{R}$, mais $\oR$, car je ne suppose plus $\oR$   closed ($\oR$ est une macro en fait, pour remettre si besoin). Open c'est bien aussi. En fait, en méditant sur les preuves,  je crois que c'est un ouvert en fait dont on a besoin, car c'est l'union d'ouverts.  Quelque part, l'idée c'est un ouvert qui sépare $\vx$ et $\REACHP(\vx)$ de $\vy$.}
We say that  $\REACHP(\vx,\vy)$ is $\epsilon$-far from being true,  if there is $\oR \subseteq X$ %(both may depend  on $\vx$ and $\vy$) 
 so that
\begin{enumerate}
\item $\vx \in \oR$,
\item $\tu f_{\epsilon}(\oR) \subseteq \oR$,
\item $\vy \not\in \oR$.
\end{enumerate}

When this holds, necessarily $\neg \REACHP(\vx,\vy)$: indeed, for all $\epsilon>0$, the set $\REACHPepsilon(\vx)=\{\vy | \REACHPepsilon(\vx,\vy)\}$  is the smallest set that satisfies $\vx \in \REACHPepsilon(\vx) $ and $\tu f_\epsilon(\REACHPepsilon(\vx))  \subseteq \REACHPepsilon(\vx)$. Consequently, as $\oR$ also satisfies these properties by the first two items, $\REACHPepsilon(\vx) \subseteq \oR$, and hence $\tu y \not\in \REACHP(\vx)$ as $\REACHP(\vx) \subseteq \REACHPepsilon(\vx) \subseteq \oR$ and $\tu y \not\in \oR$ from the third item.

In other words, $\oR$ can be seen as a \motnouv{witness} of the non-reachability of $\vy$ from $\vx$. We will say that it is \motnouv{at level $\epsilon$}.  

This provides a link to $\delta$-decidability  \cite{gao2012delta}:

\begin{proposition}[Robust $\Leftrightarrow$ reachability is true or $\epsilon$-far from being true] \label{prop28}
%Consider the hypothesis of Corollary \ref{coro6}.

We have $\REACHPomega=\REACHP$
if and only if  for all  $\vx, \vy \in \Q^{d}$, either
\begin{enumerate}
\item $\REACHP(\vx,\vy)$ is true
\item or $\REACHP(\vx,\vy)$ is false, but furthermore, there exists $\epsilon>0$ such that it is $\epsilon$-far from being true. \\
(i.e. there is a witness of it for some $\epsilon>0$ level).
\end{enumerate}
\end{proposition}

%\olivier{Il faut une preuve du coup. En fait, mon idée (basée sur ce que tu faisais au début sur le cas où cela terminait dans un certain sens),  c'est que quand on est pas dans l'item 1., i.e. $\vy$ n'est pas reachable, l'ensemble des sommets reachable dans le graphe $G_{\delta}$, pour le graphe qui l'atteste, il ne contient pas $\vy$, donc l'ouvert qui le recouvre. ET donc, on a cette propriété.}

\begin{proof}		

$\Rightarrow$:
As we said, for all $\epsilon>0$, the set $\REACHPepsilon(\vx)$ satisfies $\vx \in \REACHPepsilon(\vx) $ and $\tu f_\epsilon(\REACHPepsilon(\vx))  \subseteq \REACHPepsilon(\vx)$ (this is even the smallest set such that this holds).

	Let $\vy \in \Q^d$, let us assume that $\REACHP(\vx, \vy) = \bigcap_\epsilon 
	\REACHPomega(\vx, \vy)$ is not true. Then	 there exists $\epsilon$ such that 
	$\REACHPepsilon(\vx, \vy)$ is false, i.e. $\vy \not\in \REACHPepsilon(\vx)$.%
\olivierplusimportant{Etait: (erreurs** de toute facon) So \olivier{Pourquoi en fait?} $\distance{\vy}{\closure{\REACHPepsilon(\vx)}}>0$. 
	 Mais pourqiuoi en fait. Pourrait etre 0. Consider
	$ \epsilon' = \min({\distance{\vy}{\closure{\REACHPepsilon(\vx, \vy)}}/2, \epsilon})$, and $\overline{R}=\REACHP_{\epsilon'}(\vx)$.
	Then, $\vx \in \oR$ and from first paragraph $f_{\epsilon'}(\oR) \subseteq \oR$. Now, $\distance{\vy}{\overline{R}} =
	\distance{\vy}{\REACHP_{\epsilon'}(\vx)} 	
	\ge \distance{\vy}{\REACHPepsilon(\vx}) > \epsilon'$, using the fact that $\REACHP_{\epsilon'}(\vx) \subseteq \REACHPepsilon(\vx)$.
		Consider  $\overline{R}=\closure{\REACHPepsilon(\vx)}$.

	}
		Consider  $\oR={\REACHPepsilon(\vx)}$.
	Then, $\vx \in \oR$ and from the first paragraph $\tu f_{\epsilon}(\oR) \subseteq \oR$ and $\tu y \not\in \oR$. 	
	
$\Leftarrow$:
	When $\REACHP(\vx, \vy)$ is true, then for all $\epsilon>0$, 
	$\REACHPepsilon(\vx, \vy)$ is true, so $\REACHPomega(\vx,\vy)$ is.
  	Now, when $\REACHP(\vx, \vy)$ is false, we know by hypothesis that $\REACHP(\vx, \vy)$ is $\epsilon$-far from 
	being true for some $\epsilon>0$: there exists a set $\oR$ satisfying $\vx \in \oR$ and $\tu f_{\epsilon}(\oR) \subseteq \oR$. As 
	$\REACHPepsilon(\vx)$ is the smallest such set,  $\REACHPepsilon(\vx) \subseteq \oR$.  Now, as  $\vy \not\in \oR$, necessarily $\tu y \not\in {\REACHPepsilon(\vx)}$. 	Hence $\REACHPomega(\vx,\vy)$ is false.
%
%	If $\REACHP(\vx, \vy)$ is false and is $\epsilon$-prouvable $\epsilon$-far from 
%	being true, for some $\epsilon > 0$.
%	By definition, $\vy$ is not in the set reachable from $\vx$, and by definition
%	$\epsilon$-far from being true, there exists a set $\oR$
%	as in the definition so we have:
%	
%	\begin{itemize}
%		\item $\vx \in \overline{R}$
%		\item $f_\epsilon(\overline{R}) \subseteq \overline{R}$, so $\REACHPepsilon(\vx) \subseteq \overline{R}$
%		\item  and $d(\vy, \overline{R}) > \epsilon > 0$, so, by the previous point, $d(\vy, \REACHPepsilon(\vx)) > 0$
%	\end{itemize}
%	So $\vy \notin \REACHPepsilon(\vx)$, then, $\REACHPepsilon(\vx, \vy)$ is 
%	false, thus in that case, $\REACHPomega = \REACHP$.
%	
%	We conclude.
\end{proof}

%\olivierpourmanoninfo{Nouveau: en fait, le corollaire 28 devenu 29 tel qu'il était semble un peu de la triche, si on regarde les machines de Turing. C'est mieux de parler de l'équivalent de la terminaison. En fait, ce que tu avais fait dans tes notes au début. }

%A witness $\oR \subset X$, as well as any subset of $X$, can be encoded by its characteristic function $\chi_{\oR}: \Q^{d} \times X \to \{0,1\}$. It is said computable (over the rationals) when the latter is for rational arguments. In that case, it can be represented by the corresponding algorithm (TM).

We say that a subset $\ooR$ of $X$  is $\epsilon$-rejecting if it satisfies 2) and 3): that is to say, $\tu f_{\epsilon}(\ooR) \subseteq \ooR$,
and $\vy \not\in \ooR$.  A trajectory that reaches such a $\ooR$ will never leave it.
%
%\olivierpourmanon{Mieux comme nom que eventually decisional? En fait, c'est vraiment inspiré de ce que veut dire terminer pour les machines de Turing: on arrive sur l'ensemble des états de rejet, et on n'en bouge plus. Terminating? Mais on parle plutot de reachability pas de terminaison.}
%
A  system is  \motnouv{eventually decisional} if for all $\vx$, $\vy$, there is some $\ooR$ $\epsilon$-rejecting so that either the trajectory starting from $\vx$ reaches $\vy$, or when not, it reaches $\ooR$. 

We come back to the converse of Corollary \ref{coro6}:
from Proposition \ref{prop28}, a robust dynamical system  (i.e. $\REACHPomega=\REACHP$) is eventually decisional, by considering 
$\ooR=\oR$ for the $\oR$ given by item 2) there. 
Conversely:

\begin{lemma} 
Consider a rational  system  not robust, with $\tu f$ continuous or Lipschitz: as $\REACHP \subseteq \REACHP_{\omega}$, this means that there exist some  $\vx$ and $\vy$ with $\REACHPomega(\vx,\vy)$  but not $\REACHP(\vx,\vy)$.  The trajectory starting from $\vx$ can not reach any $\epsilon$-rejecting subset.
\end{lemma}

\begin{proof}
By contradiction, assume the trajectory starting from $\vx$ reaches some $\epsilon$-rejecting $\ooR$.  
Possibly by considering one more step, we can assume it reaches the interior of $\ooR$ for the first time at time $t$: indeed, if it reaches the frontier at $\vx^{*}$, then we know that $B(\tu f(\vx^{*}),\epsilon) \subseteq \ooR$, and $\tu f(\vx^{*})$ is in the interior of that ball. Now, from initial $\tu x$ until the position at time $t$, it remains at some positive distance of $\tu y$. As $\tu f$ is continuous or Lipschitz, the $t$-th iteration of $\tu f$ also is. So there is some $0<\epsilon'<\epsilon$ taken sufficiently small so that $\REACHP_{\epsilon'}$ intersects the interior of $\ooR$, and remains at a positive distance of $\tu y$. Once in $\ooR$, $\epsilon'$-perturbed trajectories will remain in it, since we have $\epsilon'<\epsilon$. We get $\tu y \not\in \REACHP_{\epsilon'}$, and consequently $\neg \REACHP_{\omega}(\vx,\vy)$: contradiction.
\end{proof}

%We get a converse of Corollary \ref{coro6} if some hypothesis is added.

\begin{corollary} \label{coro6r}
Consider a continuous or Lipschitz rational dynamical system.
It is robust iff it is eventually decisional. 
\end{corollary}

We can even compute the witnesses under the hypotheses of Theorem \ref{hypothesesquivontbien}.  
We say that some dynamical system is \motnouv{effectively eventually decisional} when there is an algorithm such that, given $\vx$ and $\vy$, it (terminates and) outputs such a $\ooR$ in the form of the union of rational balls.

\begin{proposition}[Reinforcement of Corollary \ref{coro6}] \label{proptrentedeux} 
Assume the hypotheses of Theorem \ref{hypothesesquivontbien}.  
If $\REACHPomega=\REACHP$ then $\REACHP$ is computable, and the system is  effectively eventually decisional.
\end{proposition}

%More precisely, we prove there is an algorithm that, for all $\vx, \vy$ with
%$\neg \REACHP(\vx,\vy)$ (terminates and) produces some computable witness of it, of the form a finite union of rational balls.

\begin{proof}
The proof of Theorem \ref{thcinq} shows that when $\REACHP_{\omega}(\vx,\vy)$ is false, then
$\REACHPepsilon(\vx,\vy)$ is false for some $\epsilon=2^{-n}$, and there is some $\delta=2^{-m}$, with some graph $G_{m}$ with some vertices $\VERTEX_{i}$ and
$\VERTEX_{j}$  with $\vx \in \VERTEX_{i}$, $\vy \in \VERTEX_{j}$ and $\neg (\VERTEX_{i }\stackrel{*}{\rightarrow_{\delta}} \VERTEX_{j})$. Denote by $\REACH{G_{m}}$ the union of the vertices $\VERTEX_{k}$ such that $\VERTEX_{i }\stackrel{*}{\rightarrow_{\delta}} \VERTEX_{k}$, for  $\vx \in \VERTEX_{i}$ in that graph. Consider $\oR= {\REACH{G_{m}}}$. This constitutes a witness at level $\delta=2^{-m}$ from the properties of the construction in that proof. 

Then $m$ can be found by testing increasing $m$ until a graph with the above properties is found. The corresponding $\oR= {\REACH{G_{m}}}$ for the first graph found will be a witness at level $\delta=2^{-m}$ from above arguments.
\end{proof}

%The converse of Corollary \ref{coro6} holds for eventualy decisional dynamical systems.
%
%\begin{proposition} \label{propconversee}
%Consider a rational discrete time dynamical system that is  eventually decisional and Lipschitz. 
%
%Then  $\REACHP_{\omega}=\REACHP$.
%\end{proposition}
%
%\begin{proof}
%Given $\vy$ and $\vy$, either $\REACHP(\vx,\vy)$ holds, and then $\REACHP_{\omega}(\vx,\vy)$ holds, or $\REACHP(\vx,\vy)$ does not hold, but in that case, we know it must reaches the corresponding $\ooR$. Possibly by considering one more step, we can assume it reaches the interior of $\ooR$ for the first time at time $t$: Indeed, if it reaches the frontier at $\vx^{*}$, then we know that $B(\tu f(\vx^{*}),\epsilon) \subseteq \ooR$, and $\tu f(\vx^{*})$ is in the interior of that ball. Now, from initial $\tu x$ until the position at time $t$, it remains at some positive distance of $\tu y$. As $\tu f$ is Lipschitz, the $t$-th iteration of $\tu f$ also is. So there is some $\epsilon>\epsilon'>0$ taken sufficiently small so that $\REACHP_{\epsilon'}$ intersects $\ooR$, and remains at a positive distance of $\tu y$. Once in $\ooR$, $\epsilon'$-perturbed trajectories will remain in it, since we have $\epsilon'<\epsilon$. We get $\tu y \not\in \REACHP_{\epsilon'}$, and consequently $\REACHP_{\omega}(\vx,\vy)$ does not hold.
%\end{proof}

An effectively eventually decisional system has its reachability relation necessarily decidable (given $\vx$ and $\vy$ compute the path until it reaches $\vy$ (then accept), or $\ooR$ (then reject)):

\begin{corollary}[Decidable $\Leftrightarrow$ robust, for  eventually decisional systems] \label{coroconversee}
Under the hypotheses of Theorem \ref{hypothesesquivontbien}, 
$\REACHP_{\omega}=\REACHP$ iff $\REACHP$ is decidable and it is effectively eventually decisional
iff it is effectively eventually decisional.
\end{corollary}

We now go to complexity issues.

Assume the dynamical system is robust, i.e. $\REACHP=\REACHPomega$. That means that for all rationals $\vx,\vy$, we have
$\REACHP(\vx,\vy) \Leftrightarrow \REACHPomega(\vx,\vy)$. Consequently,  for all rationals $\vx,\vy$,  there exists some $\varepsilon$ (depending of $\vx$, $\vy$) such that $\REACHP(\vx,\vy)$ and $\REACHPepsilon(\vx,\vy)$ have the same truth value (and unchanged by  smaller $\epsilon$).

It is then natural to quantify on the level of required robustness according to $\vx$ and $\vy$, i.e. on the value  $\varepsilon$.

As we may always assume  $\varepsilon=2^{-n}$ for some integer $n$, we write\olivierplusimportant{gain de place:\footnote{There is no true conflict of notation, as the index is here some integer number $\ge 1$, while we can always assume that in $\REACHPepsilon$, $\epsilon$ is in $[0,1)$ (even $[0,1/2)$), i.e. is not an integer.}} 
 $\REACHPn$ for $\REACHP_{2^{-n}}$, and we then introduce:
%
%Let $\length(\vx):\Q^{d} \to \N$ denote  the size of (the binary representation of) rational $\vx \in \Q^{d}$.

% \olivier{J'avais écrit ca: Let
%$\length(\vx):\Q^{d} \to \N$ denote  the size of (the binary representation of) rational $\vx \in \Q^{d}$
%
% Du coup, j'écris ca. Meme si c'est plus moche.}
% 

%\olivierpasimportant{alternative: 
%Let $\length(.)$ denote the size of a dyadic:  $\length(d)=q$ be $q$ for a dyadic $d=p/2^{q}$, and $\length(\vx)$ for a vector of dyadic $\vx$ be the sum of the length of its components.}

%We then introduce:
%
%\olivier{Ca serait clairement mieux de parler de rationnels plutot de dyadic. Mais, est-ce qu'on saurait faire proprement?}

\olivier{Est-ce un pb qu'un rationnel peut avoir plusieurs écriture? 
2/3=4/6 par exemple.... comment on la mesure}

\begin{definition}
Given some function $f: \N \to \N$, we write $\REACHPERTURBEDSPACE{\P}{f}$ for the relation defined as follows: for any rational points $\vx$ and $\vy$ the relation holds iff 
$\REACHP_{f(\length(\vx)+\length(\vy))}(\vx,\vy)$.
\end{definition}

\olivierpasimportant{polynomial time computable probablement trop fort. Mais est-ce grave docteur?}

\begin{lemma} \label{trucquivabien}
Consider a  locally Lipschitz $\Q$-computable system, with $\tu f: \Q \to \Q$ computable in polynomial time, whose domain $X$ is a closed rational box. 
For $\delta=2^{-m}$, consider the associated graph $G_{m}$ considered in the proof of Theorem \ref{thcinq}. Then $\LOOP(G_{m},\VERTEX_{i},\VERTEX_{j})$ is decidable using a space polynomial in $m$.
 \end{lemma}
 
 \begin{proof}
This graph has less than $\mathcal{O}(2^{d*m})$ vertices. The graph has a successor relation $\rightarrow_{\delta}$  computable in space polynomial in $m$.  Hence, the analysis of Corollary \ref{corosuccint}  applies, and we can determine whether $\LOOP(G_{m},\VERTEX_{i},\VERTEX_{j})$ using a space polynomial in $m$.
\end{proof}

\begin{theorem} \label{thdirectionunp}
Assume the hypotheses of Lemma \ref{trucquivabien}.
Assume $p$ is some polynomial. 
Then 
$\REACHPERTURBEDSPACE{\P}{p} \in \PSPACE$.
\end{theorem}

\begin{proof}
From the proof of Theorem \ref{thcinq}, we know that for all $n$ there exists some $m$ (depending on $n$), such that $\REACHPn(\vx,\vy)$   and $\REACHGm(\vx,\vy)$ have the same truth value, where $\REACHGm$ denotes reachability in the graph $G_{m}$.

With the hypotheses, given $\vx$ and $\vy$, we can determine whether $\REACHPERTURBEDSPACE{\P}{p}(\vx,\vy)$, by determining the truth value of $\REACHPn(\vx,\vy)$, taking $n$   polynomial in $\length(\vx)+\length(\vy)$. From the proof of Theorem \ref{thcinq}, the corresponding $m$ is polynomially related to $n$ (it is even affine in $n$). Now the analysis of  Lemma \ref{trucquivabien}, shows that the truth value of $\REACHGm(\vx,\vy)$ can be determined in space polynomial in $m$.
 \end{proof}

%\olivierpasimportant{Commenté: 	We construct a Turing machine that guesses the next point of the trajectory 
%	$\mP$, and the machine only needs to keep the points of $\mP$ 
%	in memory, which is in polynomial space in $p(length(x)+length(y))$. 
%	So we have that $\REACHPERTURBEDSPACE{\P}{p} \in NSPACE(poly(p))$, by Savitch's 
%	theorem, we can conclude.}
%%This has been established in the previous proof.

\begin{theorem}[Polynomially robust to precision $\Rightarrow$ $\PSPACE$] \label{thmaindeux} 
With same hypotheses, 
if   $\REACHP = \REACHPERTURBEDSPACE{\P}{p}$ for some polynomial $p$,
then $\REACHP \in \PSPACE$.
\end{theorem}

\begin{proof}
	We  have $\REACHPERTURBEDSPACE{\P}{p} \in \PSPACE$ by Theorem \eqref{thdirectionunp} and since $\REACHP = \REACHPERTURBEDSPACE{\P}{p}$, then $\REACHP \in \PSPACE$.
\end{proof}

Actually, this is even a characterization of $\PSPACE$ (if one prefers: Reachability is $\PSPACE$-complete for $\Q$-computable poly-time computational,  polynomial robust to precision).

%\olivier{On est d'accord?}

\begin{theorem}[Polynomially robust to precision $\Leftrightarrow$ $\PSPACE$]  \label{recipir}
Any  $\PSPACE$ language is reducible to the reachability relation of PAM with $\REACHP = \REACHPERTURBEDSPACE{\P}{p}$ for some polynomial $p$.
\end{theorem}

\begin{proof}
Let $L \in \PSPACE$. There is a TM $\M$ with $L(\M)=L$ that works in polynomial space $q(\cdot)$. Its step-by-step emulation considered in Theorem \ref{chosemachin}, using $\gammacompact$ is done using a precision $\mathcal{O}(2^{-q(n)})$ on words of length $n$. The obtained system satisfies $\REACHPERTURBEDSPACE{\P}{q + \mathcal{O}(1)}=\REACHP$ from  the properties of the emulation. 
In other words, this comes from the fact that the involved emulation preserves robustness of TMs.
\end{proof}

%
%\olivierpasimportant{paragraphe qui ne veut rien dire.}
%From the proofs, in the two previous theorems, the hypothesis that the domain $X$ of $\tu f$ is a closed rational box can be weaken into
%the fact that the discrete time dynamical system admits a suitable polynomial space computable abstraction.
%
%\olivierpasimportant{Suis-je en train d'arnaquer?}
Assuming the same hypotheses as in Theorem \ref{thmaindeux}, when $\REACHP = \REACHPERTURBEDSPACE{\P}{p}$ for some polynomial $p$, we also see that we can determine a witness of the fact that $\neg \REACHP(\vx,\vy)$ in polynomial space (using a suitable representation of it).

%\section{General discrete time dynamical systems}

\olivierplusimportant{subsection{Basic definitions}}

\subsection{The case of computable systems}
\label{sub:parttwo}

%\olivierpourmanoninfo{Prise de tete: besoin de distinguer si on atteind boule ouverte, ou fermée}

We consider now the case of general (possibly non-rational) discrete time dynamical systems.  In that case, $\tu f$ may take some non-rational values, and we  need to talk about computability for functions over the reals.  A  system is said computable  if the function $\tu f: \R^{d} \to \R^{d}$ is in computable analysis (CA).
 %
%That means that we assume from now on some basic knowledge of computable analysis (see e.g. \cite{Wei00} or \cite{brattka2008tutorial}). 

A crash course on CA can be found in the appendix (see e.g.  \cite{Wei00} or \cite{brattka2008tutorial}), but in very short:  \olivierplusimportant{Basically, the idea behind classical computability and complexity is to fix some representations of objects (such as graphs, integers, etc, \dots) using finite words  over some finite alphabet, say $\Sigma=\{0,1\}$, and to say that such an object is computable when such a representation can be produced using a Turing machine. The aim of computable analysis is to be able to talk also about objects such as real numbers, functions over the reals, closed subsets, compacts subsets, \dots, which cannot be represented by finite words over $\Sigma$ (a clear reason for it is that such words are countable while $\R$ for e.g. is not).  However, they can be however represented by some infinite words over $\Sigma$, and the idea is to fix such representations for these various objects, called \emph{names}, with suitable computable properties.  }
a name for a point $\vx \in \R^{d}$ is a sequence $(I_{n})$ of nested open rational balls  with $I_{n+1} \subseteq I_n$ for all $n \in \mathbb{N}$ and $\{x\}=\bigcap_{n \in \mathbb{N}} I_n$. A name for a  function $\tu f: \R^{d} \to \R^{d'}$ is a list of all pairs of open rational balls $(I, J)$ such that $\tu f(\closure{I}) \subseteq J$. A name for  a closed set $F$ is a sequence $(I_{n})$ of all open rational balls such that $\closure{I_{n}} \cap F=\emptyset$ and a  sequence $(J_{n})$ of all open rational balls such that $J_{n} \cap F \neq \emptyset$. A name for  a compact $K$ is a name of $F$ as a closed set, and an integer  $L$ such that $K \subset B(0,L)$.
%\olivier{Ca ou autre chose. Verifier que pas fumé, et que c'est bien ca dans la littérature.}
%\end{itemize}
%
  All these names can  be encoded as infinite sequences of symbols. The notion of computability involved is the one of Type 2 Turing machines, that is to say machines possibly working over infinite tapes, and outputting their results in possibly write-only output tapes.
Then: 
a point $\vx \in \R^{d}$ is computable if it has a computable name. And similarly for defining the concept of computable function, computable closed set, or computable compact: we mean for example, that a closed set is computable if it has a computable name (and a compact is computable consequently also as a closed set). 
In particular this concept of computability for compacts implies the property discussed after Theorem \ref{Rcorec} (see \cite[Lemma 5.2.5]{Wei00} for a proof). %, and the latter theorem is hence even effective in a name of the compact. 
\olivierplusimportant{	  (This follows from \cite[Lemma 5.2.5]{Wei00}, that proves that  $\kappa_{\mathrm{mc}} \equiv \kappa$ for the representations defined there for compact sets, and the observation that being computable is equivalent to having a computable name (See Section \ref{Sec:analysecalculable})).  Furthermore, }
If $Y$ and $Z$ are spaces with an associated naming system, then an operator $f: Y \rightarrow Z$ is said computable if there is a computable function which associates each name of $y \in Y$ to a name of $f(y) \in Z$.

%\olivierpourmanoninfo{NEW! NEW! En fait, révélation: je crois que tout est dans cette chose.}

From the model of CA, given the name\footnote{Even assuming it is computable.} of $\tu f$, and (even for) some rational $\vx$ and $\vy$, this is impossible to tell effectively if  $\tu f(\tu \vx)=\vy$ in the general case. Consequently, given some rational ball $B(\tu y,\delta)$, we have to forbid ``frontier reachability'', that is to say the case where $B(\tu y,\delta)$ would not be reachable, but its frontier is, i.e. $\cB(\tu y,\delta)-B(\tu y,\delta)$ is reachable. A natural question arises then: given some ball such that either $B(\tu y,\delta)$ is reachable (that case implies that $\cB(\tu y,\delta)$ is), or such that $\cB(\tu y,\delta)$ is not, decide which possibility holds.  %\olivier{Meilleur nom pour cette chose} 
We call this the \motnouv{ball (decision) problem}. Of course, from the above definitions, when $\REACHP(\vx)$ is a closed set, $\REACHP(\vx)$  is a computable closed set iff the associated ball problem is algorithmically solvable. 

For a computable system, the ball decision problem is c.e.:  simulate the evolution of the system starting from $\vx$ until step $T$, with increasing precision and $T$, until one finds the guarantee that the position $\vx_{T}$ at time $T$ remains in $B(\vy,\delta')$ for some $\delta'<\delta$. This works: if the ball is indeed reachable, this will terminate by eventually computing a sufficient approximation of  the corresponding $\vx_{T}$,  and conversely it can't terminate without guaranteeing reachability.   The ball problem is of course not co-c.e. in general. 

To a discrete time system,  we can also associate its reachability relation $\REACHP(\cdot, \cdot, \cdot)$ over  
$\Q^d\times\Q^{d}\times\N$. Namely, for two rational points $\mathbf{x}$ and $\mathbf{y}$ and rational $0<\eta=2^{-p}$, encoded by the integer $p$, the relation $\REACHP(\mathbf{x}, \mathbf{y},p)$  holds iff there exists a trajectory of $\mP$ from $\vx$ to $\cB(\vy,\eta)$. We define $\REACHPepsilon$ similarly, and $\REACHPomega = \bigcap_{\epsilon} \REACHPepsilon$. This relation encodes  reachability with arbitrarily small perturbing noise to some closed ball.

%\olivier{enleve: CHANGE ENCORE D'AVIS.
%We define  $\REACHPomega(\vx, \vy, p)$ iff $\forall \varepsilon>0$,  there exists a trajectory of $\mP_{\epsilon}$ from $\vx$ to $\cB(\vy,\eta)$. This relation encodes  reachability with arbitrarily small perturbing noise (to the \emph{closed} ball).
%From definitions, we have:
%}

\begin{lemma}
for any $0<\varepsilon_2<\varepsilon_1$ and any $\mathbf{x}$ and $\mathbf{y}$, $\eta$, the following implications hold:
$\REACHP(\mathbf{x}, \mathbf{y},p)$ $ \Rightarrow \REACHPomega(\mathbf{x}, \mathbf{y},p) \Rightarrow \REACHP_{\varepsilon_2}(\mathbf{x}, \mathbf{y},p) \Rightarrow \REACHP_{\varepsilon_1}(\mathbf{x}, \mathbf{y},p)$.
\end{lemma}

%\olivierpourmanoninfo{Encore changé là: j'y arriverai un jour}
%

Given  $\vx$, and $0<\varepsilon_2<\varepsilon_1$, we have $\REACHP(\vx) \subseteq \REACHP_{\epsilon_{2}}(\vx) \subseteq \closure{\REACHP_{\epsilon_{2}}(\vx)} \subseteq \REACHP_{\epsilon_{1}}(\vx) \subseteq \closure{\REACHP_{\epsilon_{1}}(\vx)}$. Consequently, $\REACHPomega(\vx)=\bigcap_{\epsilon>0} \REACHPepsilon(\vx) =  \bigcap_{\epsilon>0} \closure{\REACHPepsilon(\vx)}$   is a closed set. 

%\olivier{En fait là, il ya un pb. Pas sur le th 38. Son utilisation. J'y reflechis.}

%\olivierpourmanoninfo{Meme chose: Oublier ce qui suit. Le blemme était dans le truc d'aant en fait.
%\olivierpourmanoninfo{La prise de tete vient de cette prise de tete sur laquelle tu as passé beaucoup d'energie. Décidément une vraie vraie prise de tete, et un gros résultat (voir le plus fondamental). e pensais qu'on pouvait s'en tirer pour parler de $\eta$ ensuite, mais en fait, ca fait sauter mes neurones pleins de fois, mais je pense que à condition que l'on est pas en train d'atteindre un point qui est sur la boule fermée et pas sur la boule ouverte. }}

\begin{theoremd} 
[Perturbed reachability is co-r.e.]  \label{pRcorec} \label{thprisedetete}

Consider a  locally Lipschitz  computable system whose domain $X$ is a computable compact.
 $\REACHPomega(\mathbf{x}, \mathbf{y},p) \subseteq \Q^d \times \Q^{d} \times \N$ is in  $\Pi_1^0$.
\end{theoremd}

%\olivier{Il faut modifier ce texte pour répercuter cette histoire de boule fermée. Pour l'instant un peu arnaque.}

This statement have similarities with \cite[Theorem 13]{BGHRobust10}: the result established their is about the language accepted by a system, but with very strong hypotheses on termination compared to ours which make their analysis really simpler.

\begin{proof}
As $\tu f$ is locally Lipschitz, and $X$ is compact, we know that $\tu f$ is Lipschitz: there exists some $L>0$ so that $\distance{\tu f(\vx)}{\tu f(\vy)} \le L \cdot \distance{\vx}{\vy}$. 

	For every $\delta=2^{-m}$, $m \in \N$, we associate some  graph $G_{m}=(V_{\delta},\rightarrow_{\delta})$: the vertices, denoted  $(\VERTEX_{i})_{i}$, of this graph correspond  to some finite covering of compact $X$ by rational open balls $\VERTEX_{i}=B(\vx_{i},\delta_{i})$  of  radius $\delta_{i} <\delta$. 

There is an edge from $\VERTEX_i$ to $\VERTEX_j$ in this graph, that is to say $\VERTEX_{i} \rightarrow_{\delta} \VERTEX_{j}$,  iff $B(\tu f_{i},(L+2)\delta) \cap \VERTEX_{j} \neq \emptyset$, given some rational $\tu f_{i}$ given by some (computed) $\delta$-approximation of $\tu f(\vx_i)$, i.e. $\tu f_{i}$ such that $\tu f(\vx_{i}) \in B(\tu f_{i}, \delta)$.  

This is done to guarantee to cover $B(\tu f(\vx_i), (L+1) \delta)$.

%$\distance{\tu f(\vx_{i})}{\vy}<(L+1)\delta$ for some $\vy \in \VERTEX_{j}$, and $f(\vx_i) \in \R$.
%$B(f, (L+1)\delta) \cap \VERTEX_{j} \neq \emptyset$, with $f$ a $\delta$-approximation of $f(\vx_i=)$, which is equivalent to verifying if $B(f(\vx_i), (L+2) \delta)  \cap \VERTEX_j) \neq \emptyset$.
%\olivierpourmanon{voir commentaire dans th 26: ca peut aussi s'écrire autrement dans th 26. Donc ici aussi adaptable je pense.} 
%`This implies 
%	$B(\tu f(\vx_{i}),(L+2)\delta) \cap \VERTEX_{j} \neq \emptyset$. 
%	
	 As we assumed compact $X$ to computable, such a graph can be effectively obtained from $m$, by computing suitable approximation $\tu f_{i}$ of the $\tu f(\vx_{i})$'s at precision $\delta$.
	 
	 We write, as expected, $\REACHP(\vx,\vy)$ if there is a trajectory from $\vx$ to $\vy$, allowing $\vy$ to be some real point (and similarly for 
	 $\REACHPepsilon(\vx,\vy)$.

\begin{itemize}
	\item Claim 1: 
	 assume $\REACHPepsilon(\vx,\vy)$ with $\vx \in \VERTEX_{i}$ for $\epsilon=2^{-n}$. Then $\VERTEX_{i} {\rightarrow_{\epsilon}} \VERTEX_{j}$ for all  $\VERTEX_{j}$ with $\vy \in \VERTEX_{j}$.
	 
	 This basically holds as the graph for $\delta=\epsilon$ is made to always  have more trajectories/behaviours than $\REACHPepsilon$. 
	
%	 If $\vy \in \tu f_{\epsilon}(\vx)$, then $\distance{\tu f(\vx_{i})}{\vy} \le \distance{\tu f(\vx_{i})}{\tu f(\vx)} + \distance{\tu f(\vx)}{\vy} < L \distance{\vx_{i}}{\vx} + \epsilon \le L \epsilon + \epsilon = (L+1) \epsilon$, and hence there is an edge from $\VERTEX_{i} {\rightarrow_{\epsilon}} \VERTEX_{j}$ to any $\VERTEX_{j}$ containing $\vy$ by definition of the graph.	
	
	\begin{proof} If $\vy \in \tu f_{\epsilon}(\vx)$, 
	%and if we consider $\tu f_{i}$ as  the $\epsilon$ approximation of $\tu f(\vx_{i})$, 
	then 
	$\distance{\tu f_{i}}{\tu y} \le \distance{\tu f_{i}}{\tu f(\vx_{i})} + \distance{\tu f(\vx_{i})}{\tu f(\vx)} + \distance{\tu f(\vx)}{\vy} < \epsilon + L \distance{\vx_{i}}{\vx} + \epsilon = (L+2) \epsilon$, and hence there is an edge from $\VERTEX_{i} {\rightarrow_{\epsilon}} \VERTEX_{j}$ to any $\VERTEX_{j}$ containing $\vy$ by definition of the graph.
%
%	
%	
%	$\distance{\tu f(\vx_{i})}{\vy} \le \distance{\tu f(\vx_{i})}{\tu f(\vx)} + \distance{\tu f(\vx)}{\vy} < L \distance{\vx_{i}}{\vx} + \epsilon \le L \epsilon + \epsilon = (L+1) \epsilon$, and hence there is an edge from $\VERTEX_{i} {\rightarrow_{\epsilon}} \VERTEX_{j}$ to any $\VERTEX_{j}$ containing $\vy$ by definition of the graph.
	\end{proof}
	
	\item Claim 2: for any $\epsilon=2^{-n}$, there is  some $\delta=2^{-m}$ so that if we have $\VERTEX_{i} \stackrel{*}{\rightarrow_{\delta}} \VERTEX_{j}$ then  $\REACHPepsilon(\vx,\vy)$ whenever $\vx \in \VERTEX_{i}$, $\vy \in \VERTEX_{j}$.
	
	\begin{proof} Consider $\delta=2^{-m}$ with $\delta <\epsilon/(2L+4)$: assume  $\VERTEX_{i=i_{0}} {\rightarrow_{\delta}} \VERTEX_{i_{1}} \dots {\rightarrow_{\delta}}  \VERTEX_{i_{t}=j}$ with $\vx \in \VERTEX_{i}$, $\vy \in \VERTEX_{j}$. 

	Assume by contradiction that $\neg  \REACHPepsilon(\vx,\vy)$, and let  $\ell$ be the least index such that 
	%$\REACHPepsilon(\vx,\overline x)$ for all $\overline x \in  \VERTEX_{i_{\ell}}$ but but 
	$\neg \REACHPepsilon(\vx,\overline \vz)$ for some $\overline \vz \in \VERTEX_{i_{\ell+1}}$.

		As 	 $\VERTEX_{i_{\ell}} {\rightarrow_{\delta}} \VERTEX_{i_{\ell+1}}$ there is some $\overline {\tu y} \in  \VERTEX_{i_{\ell+1}}$ with 
		$\distance{\tu f_{i_{\ell}}}{\overline{\vy}}<(L+2) \delta$ with $\distance{\tu f_{i_{\ell}}}{\tu f(x_{i_{\ell}})} < \delta$. 
		%This implies  $\distance{\tu f(x_{i_{\ell}})}{\overline{\vy}}<(L+3) \delta$.   
		Take $\overline{\tu z} \in \VERTEX_{i_{\ell+1}}$. 
		
			If $\ell=0$, then $\distance{\tu f(\vx)}{\overline{\tu z}} \le \distance{\tu f(\vx)}{\tu f(\vx_{i_{\ell}})}+\distance{\tu f(\vx_{i_{\ell}})}{\tu f_{i_{\ell}}}  +\distance{\tu f_{i_{\ell}}}{\overline{\vy}} + \distance{\overline{\vy}}{\overline{\vz}} < L \delta + \delta +  (L+2) \delta + \delta = (2 L + 4) \delta < \epsilon$, and hence
	$\REACHPepsilon(\vx,\overline{\vz})$: contradiction.
		
		If $\ell>0$, as $\ell$ is the least index with the above property, $\REACHPepsilon(\vx,\vx_{i_{\ell}})$. But then 
	$\distance{\tu f(\vx_{i_{\ell}})}{\overline{\tu z}} \le \distance{\tu f(\vx_{i_{\ell}})}{\tu f_{i_{\ell}}}  + \distance{\tu f_{i_{\ell}}}{\overline{\vy}} + \distance{\overline{\vy}}{\overline{\vz}}
	< \delta + (L+2) \delta + \delta < (2 L + 4) \delta < \epsilon$. 
		
	And hence, $\REACHPepsilon(\vx_{i_{\ell}},\overline{\vz})$, and since we have $\REACHPepsilon(\vx,\vx_{i_{\ell}})$, we get $\REACHPepsilon(\vx,\overline \vz)$ and a contradiction.
	\end{proof}

\end{itemize}

	 	That is, Claim 2 says that $\neg \REACHPepsilon(\vx,\vy)$ implies $\neg (\VERTEX_{i} \stackrel{*}{\rightarrow_{\delta}} \VERTEX_{j})$ whenever $\vx \in \VERTEX_{i}$, $\vy \in \VERTEX_{j}$, for the corresponding $\delta$.

	 From the two above items, $\REACHPomega(\vx,\vy)$ holds iff for all $\delta=2^{-m}$, we have $\VERTEX_{i} \stackrel{*}{\rightarrow_{\delta}} \VERTEX_{j}$, for all  $\VERTEX_{i}$, $\VERTEX_{j}$ with  $\vx \in \VERTEX_{i}$, $\vy \in \VERTEX_{j}$.
	 
	 	 If one prefers, $\neg \REACHPomega(\vx,\vy)$ holds iff for some $\delta=2^{-m}$, $\neg (\VERTEX_{i }\stackrel{*}{\rightarrow_{\delta}} \VERTEX_{j})$ for some $\VERTEX_{i}$, $\VERTEX_{j}$ with $\vx \in \VERTEX_{i}$, $\vy \in \VERTEX_{j}$.
		 
%		 \olivierpourmanoninfo{JE CROIS QUE CELA MARCHE BIEN. ET C'EST CE QUE TU ESSAYAIS DE FAIRE DEPUIS LE DEBUT... j'ai l'impression d'une certaine fa\c con}
%		 Then $\neg \REACHPomega(\vx,\vy,p)$ holds iff for some $\delta=2^{-m}$, $\neg (\VERTEX_{i }\stackrel{*}{\rightarrow_{\delta}} \VERTEX_{j})$ for any $\VERTEX_{i}$, $\VERTEX_{j}$ with $\vx \in \VERTEX_{i}$, $\VERTEX_{j} \cap \cB(\vy,2^{-p}) \neq \emptyset$.
		 
		Then:

		 \begin{itemize}
		 \item Claim* (compactness argument):  Given a ball  $B(\vy,\eta)$, we have that $\cB(\vy,\eta) \cap \REACHPomega(\vx)=\emptyset$ iff $\cB(\vy,\eta) \cap \closure{\REACHPepsilon(\vx)}=\emptyset$ for some $\epsilon>0$.
		 \end{itemize}

		 \begin{proof} 
		 		$\Leftarrow$ (easy direction):  If $\cB(\vy,\eta) \cap \REACHPepsilon(\vx)= \emptyset$ for some $\epsilon>0$, we cannot have $\cB(\vy,\eta) \cap \REACHPomega(\vx) \neq \emptyset$, as it would contain a point that would necessarily be in $\cB(\vy,\eta) \cap \REACHPepsilon(\vx)$.
				 
				 $\Rightarrow$ (compactness argument): 
				 		  %\olivier{Preuve version ouverte? Serait mieux}
						Assume  $\cB(\vy,\eta) \cap \REACHPomega(\vx)=\emptyset$. Since $\REACHPomega(\vx)=\bigcap_{\epsilon>0} \closure{\REACHPepsilon(\vx)}$, this means that $\bigcup_{n \in \N} (\closure{\REACHP_{{2^{-n}}}(\vx)})^{c}$ is some covering of $\cB(\vy,\eta)$. As $\cB(\vy,\eta)$ is closed and bounded, it is compact. Consequently, from the covering can be extracted some finite covering. Consequently, $\bigcup_{n \ne n_{0}} (\closure{\REACHP_{{2^{-n}}}(\vx)})^{c}= (\closure{\REACHP_{{2^{-n_{0}}}}(\vx)})^c$ for some $n_{0}$ is a covering of $\cB(\vy,\eta)$.  In other words, 	$\cB(\vy,\eta) \cap \closure{\REACHPepsilon(\vx)}=\emptyset$ for $\epsilon=2^{-n_{0}}$. 	This proves the direction from left to  right.
					
		 \end{proof}

		 Consequently, we basically can use arguments really similar to those of  the proof of Theorem \ref{thcinq}, where the role played by $\tu y$ is now played by $\cB(\vy,\eta)$. With more details:

	\olivierpasimportant{Dans theorem d'avant, cétait: 
	
		 \olivierpasimportant{ Pour aider:
 etape 1	 
	 $\REACHPomega(\vx,\vy)$ holds iff for all $\delta=2^{-m}$, we have $\VERTEX_{i} \stackrel{*}{\rightarrow_{\delta}} \VERTEX_{j}$, for all  $\VERTEX_{i}$, $\VERTEX_{j}$ with  $\vx \in \VERTEX_{i}$, $\vy \in \VERTEX_{j}$.

 Etape 1'
	 	 If one prefers, 
}

	}
	
		 $\REACHPomega(\vx,\vy,\eta)$ holds iff for all $\delta=2^{-m}$, we have $\VERTEX_{i} \stackrel{*}{\rightarrow_{\delta}} \VERTEX_{j}$, for all  $\VERTEX_{i}$, $\VERTEX_{j}$ with  $\vx \in \VERTEX_{i}$, $\VERTEX_{j} \cap \cB(\vy,2^{-p})  \neq\emptyset$. 

\olivierpasimportant{Preuve de la phrase plus haut:}

	The direction from left to right is clear from Claim $1$. \olivierpasimportant{En fait, aps besoin  un claim 1 sur fermé. On parle bien de y dans le fermé.No bug.} Conversely, assume that for all $\delta=2^{-m}$, we have $\VERTEX_{i} \stackrel{*}{\rightarrow_{\delta}} \VERTEX_{j}$, for all  $\VERTEX_{i}$, $\VERTEX_{j}$ with  $\vx \in \VERTEX_{i}$, $\VERTEX_{j} \cap \cB(\vy,2^{-p})  \neq\emptyset$. Assume by contradiction that $\neg \REACHPomega(\vx,\vy,\eta)$.
From Claim*,  we know that $\cB(\vy,\eta) \cap \closure{\REACHPepsilon(\vx)}=\emptyset$ for some $\epsilon>0$. In particular $\cB(\vy,\eta) \cap {\REACHPepsilon(\vx)}=\emptyset$. Then for the corresponding $\delta=2^{-m}$ from Claim $2$, we cannot have $\VERTEX_{i} \stackrel{*}{\rightarrow_{\delta}} \VERTEX_{j}$, for any  $\VERTEX_{i}$, $\VERTEX_{j}$ with  $\vx \in \VERTEX_{i}$, $\VERTEX_{j} \cap \cB(\vy,2^{-p})  \neq\emptyset$. 
	This proves the direction from right to left.

\olivierpasimportant{VIEUX, VIEUX, HORS SUJET NORMALEMENT:
		 If there is some $\cB(\vy,2^{-p})$ that is not reachable from $\vx$. By claim $2$, for the corresponding $\delta$ by the $\epsilon$-perturbed dynamics, then any $\VERTEX_{i}$, $\VERTEX_{j}$ with $\vx \in \VERTEX_{i}$, $\VERTEX_{j} \cap B(\vy,2^{-p}) \neq \emptyset$ is such that $\neg(\VERTEX_{i }\stackrel{*}{\rightarrow_{\delta}} \VERTEX_{j})$.

	 	Conversely, for some $\delta=2^{-m}$, $\neg (\VERTEX_{i }\stackrel{*}{\rightarrow_{\delta}} \VERTEX_{j})$ for any $\VERTEX_{i}$, $\VERTEX_{j}$ with $\vx \in \VERTEX_{i}$, $\VERTEX_{j} \cap B(\vy,2^{-p}) \neq \emptyset$, then $\neg \REACHP_{\delta}(\vx,\vy,p)$ by Claim 1. 
		 }
		 
		 \olivierpasimportant{Donc, on reprend}
		 If one prefers, 
		$\neg \REACHPomega(\vx,\vy,\eta)$ holds  iff for some integer $m$, the following property $P_{m}$ holds: $\LOOP(G_{m},\VERTEX_{i},\VERTEX_{j})$ for any  $\VERTEX_{i}$ and $\VERTEX_{j}$ with $\vx \in \VERTEX_{i}$ and $\VERTEX_{j} \cap \cB(\vy,2^{-p})  \neq\emptyset$. 
	 
	The latter property is computably enumerable, as it corresponds to a  union of decidable sets (uniform in $m$), as 	the property $P_{m}$  is a decidable property over finite graph $G_{m}$.
\end{proof}

%\olivierpourmanoninfo{NEW! NEW! FORMULATION}

\begin{corollary}[Robust $\Rightarrow$ decidable]   \label{coro6p}
Assume the hypotheses of Theorem \ref{pRcorec}.   Assume that for all rational $\vx$, $\REACHP(\vx)$ is closed, and
$\REACHP(\vx)=\REACHPomega(\vx)$.
Then the ball decision problem is decidable.
%\olivier{On peut le dire mieux?}
%: for all rational $\vx$, $\REACHP(\vx)$ is a computable closed set. 
%$\REACHP= \REACHPomega$ over  $\Q^{d} \times \Q^{d} \times \N$.
%That is, for all rationals $\vx, \vy$ and integer $p$,
%we have $\REACHP(\vx,\vy,p)= \REACHPomega(\vx,\vy,p)$.

%Then $\REACHP \subseteq \Q^{p} \times \Q^{p} \times \N$ is decidable.
\end{corollary}

%\olivier{On est bien d'accord?}
\begin{proof}
Given some instance $B(\tu y,\delta)$ of the ball problem, run in parallel the c.e. algorithm for it (and when its termination is detected, accepts) and the c.e. algorithm for $\left(\REACHP(\vx)\right)^{c}=\left(\REACHPomega(\vx)\right)^{c}$ (and when its termination is detected, rejects). 
%
%
% an enumeration of the balls $B(\vy_{i},2^{-p_{i}})$ with $\neg \REACHPomega(\vx,\vy,p)$, until a cover of $\cB(\tu y,\delta))$ with these balls is found (and when this is found, rejects). 
%Clearly, if it terminates, its answer is correct.  Now, if it does not terminate, this means $\cB(\tu y,\delta) \cap \REACHP(\vx) \neq \emptyset$ (as no cover can separate the two), and hence that the instance was not a correct instance: that $B(\tu y,\delta)$ is ``frontier reachable''.
%%
%%The results follows, observing that this procedure is even effective in a name of $\vx$. 
\end{proof}

\begin{definition}
Given some function $f: \N \to \N$, we write $\REACHPERTURBEDSPACE{\P}{f}$ as: for two rational points $\mathbf{x}$ and $\mathbf{y}$, and $p$, the relation holds iff 
$\REACHP_{f(\length(\vx)+\length(\vy)+p)}(\vx,\vy,p)$.
\end{definition}

We need also the concept of polynomial (poly.) time computable function: see \cite{Ko91}. In short, a quickly converging name of $\vx \in \R^{d}$ is a name  of $\vx$, with $I_{n}$ of radius $<2^{-n}$. Then $\tu f: \R^{d} \to \R^{d'}$ is computable in poly. time, if there is some oracle TM $M$, such that, for all $\vx$, given any fast converging name of $\vx$ as oracle, given $n$, $M$ produces some open rational ball of radius $<2^{-n}$ containing $\tu f(\vx)$,  in a time poly. in $n$.

%\olivier{En fait, on utilise la notion de calculable en temps polynomial de computable analysis. Pas défini.}

%\olivierpasimportant{C'est pas ce concept en fait, mais celui de computable analysis.
%We say (as usual) that a function $\tu f: \Q^{d} \to \Q^{d}$ is computable in polynomial space if $\tu f(\vx)$ can be produced
%in a space polynomial in $\length(\vx)$. 
%}

%\olivier{besoin du concept de computable in polynmial space du coup.}

\begin{theorem} \label{thdirectionunpbis}  \label{ab:ab:ab}
Consider a  locally Lipschitz system, with $\tu f$ polynomial time computable, whose domain $X$ is a closed rational box. 
Then $\REACHPERTURBEDSPACE{\P}{p} \subseteq \Q^{d}\times \Q^{d} \times \N \in \PSPACE$.
\end{theorem}

\olivierpasimportant{Reperctuer du coup}
\begin{proof}
The proof of Theorem \ref{pRcorec} (similar to the one of Theorem \ref{Rcorec})   shows that when $\REACHP_{\omega}(\vx,\vy,p)$ is false, then
$\REACHPepsilon(\vx,\vy,p)$ is false for some $\epsilon=2^{-n}$.
With the hypotheses, given $\vx$, $\vy$ and $p$, we can take $n$   polynomial in $\length(\vx)+\length(\vy)+p$. From the proof, the corresponding $m$ is polynomially related to $n$ (it is even affine in $n$). Now an analysis similar to the one of  Lemma \ref{trucquivabien}, shows that the truth value of $\REACHGm(\vx,\vy,p)$ can be determined in space polynomial in $m$.
\end{proof}

%\olivierpasimportant{Ok que ca roule exactemnet comme avant? (enfin, c'est lié au theorem laissé à Manon, donc juste vérifier).}

%\begin{proofmanon}[\textcolor{red}{NOUVEAU}]
%	Let $\mathbf{x}, \mathbf{y} \in \Q$. We just have to construct a Turing machine 
%	that guesses the next point of the trajectory with the right property, with 
%	polynomial, non-deterministically. We use Savitch theorem to conclude. 
%\end{proofmanon}

\olivierplusimportant{Du coup, faut parler de boules, e,n fait, non?}

\begin{theorem}[Polynomially robust to precision $\Rightarrow$ $\PSPACE$] \label{thmaindeuxp} 
Consider the hypotheses of Theorem \ref{ab:ab:ab}.  
%Assume that $\REACHP$ is closed, or more generally that $\bd(\REACHP(\vx)) \cap \Q^{d} = \emptyset$.
%
Assume that for all rational $\vx$, $\REACHP(\vx)$ is closed, and
$\REACHP(\vx)=\REACHPERTURBEDSPACE{\P}{p}$ for some polynomial $p$.
Then the ball decision problem is in $\PSPACE$. 
%
%If 
%$\REACHP = \REACHPERTURBEDSPACE{\P}{p}$ for some polynomial $p$, 
%then $\REACHP \in \PSPACE$.
\end{theorem}

\begin{proof}
 	 We  have $\REACHPERTURBEDSPACE{\P}{p} \in \PSPACE$ by Theorem \ref{thdirectionunpbis} and since $\REACHP = \REACHPERTURBEDSPACE{\P}{p}$, then $\REACHP \in \PSPACE$.
\end{proof}

\section{Relating robustness to drawability}
\label{Sec:analysecalculable}

\olivierpasimportant{Enlevé: 
Notice that, even if arguments are similar in the end, the previous analysis of $\Q$-rational systems (Section \ref{sub:partone}) is not a special case of the analysis done for computable systems (Section \ref{sub:parttwo}).  Indeed, a function computable in computable analysis is necessarily continuous, while in the analysis of $\Q$-rational systems, functions could be discontinuous.}
\olivierpasimportant{Enlevé:  (and PAM systems, taken several time as examples, do not assume any form of continuity). For $\Q$-rational systems, classical computability over the rationals was involved, while for the second, this was computable analysis.  }

We can go even further, and go to geometric properties: in the previous sections, we associated to every discrete time dynamical system a reachability relation over the rationals. 
%\olivier{Enlevé pour gain de place:  (to be more precise and correct over $\Q^{d} \times \Q^{d}$ or $\Q^{d} \times \Q^{d} \times \N$)}  
 But we could also see it as a relation over the reals, and use the framework of CA, talking about subsets of $\R^{d}\times \R^{d}$. 
%
%Indeed, as we did for PAMs, to each dynamical system $\mP$, we associate  its reachability relation $\REACHP(\cdot, \cdot)$ now seen on $\R^d\times \R^{d}$. 
\olivierpasimportant{Superflus je pense: Namely, for two  points $\mathbf{x}$ and $\mathbf{y}$ the relation $\REACHP(\mathbf{x}, \mathbf{y})$ holds iff there exists a trajectory of $\mP$ from $\vx$ to $\vy$.}
%
%\olivier{Expliquer le  sens à ma formulation de effectively pour truc semi-calculable.}
%
%Perturbed versions are then defined as in previous sections.
%We first recall the very basics of computable analysis (see e.g. \cite{Wei00}).
%\olivierpasimportant{attention à names vs computable dans chaque formulation de mes/ces trucs.}
%
%Notice that the machine must behave correctly only for every name of every $x \in \operatorname{dom}(f)$, for other input sequences it may behave arbitrarily.
A closed subset of $\R^{d}$ is called  c.e. closed if one can effectively enumerate the rational open balls intersecting it.
From the statements of \cite{Wei00}, the following holds:
\begin{theoremd}	 \label{baba}  \label{thbaba} 
Consider a  computable discrete time system $\mP$ whose domain  is a computable compact.

For all $\vx$,  $\closure{\REACHP(x)} \subseteq \R^{d}$  is a c.e. closed subset.
\end{theoremd}

\begin{proof}
We do the proof for the case of discrete time dynamical system. Write $\REACH{\P,T}(\mathbf{x}, \mathbf{y})$  iff there exists a trajectory of $\mP$ from $\vx$ to $\vy$ in less than $T$ steps. We can write $\REACH{\P,0}(\mathbf{x}, \mathbf{y})$ as $\{(\vx,\vy)| \vx=\vy\}$, which is a computable closed subset %(\mylabelbt{Examples 5.6.1},
\mylabelwei{Example 5.1.3}). We can then also write  $\REACH{\P,T+1}(\mathbf{x}, .)= \tu F(\REACH{\P,T}(\mathbf{x}, .))$ where $\tu F(K):=K \cup \tu f(K)$. As $\tu f$ is computable, we know it is continuous, and by induction on $T$, $\REACH{\P,T+1}(\mathbf{x}, \vy)$ is a compact: indeed, as $\tu f$ is computable, it is continuous, and as $K$ is a closed subset living in a compact by induction, it is compact, and the image of a compact by some continuous function is compact. As $\tu f$ is computable, and $K$ is compact, we know that $\tu f(K)$ is computable (\mylabelwei{Theorem 6.2.4}), and hence also $\tu F(K)$ (\mylabelwei{Theorem 5.1.13}). And then by induction on $T$, that $\REACH{\P,T}(\mathbf{x}, \mathbf{y})$ is a closed computable subset. A computable closed set is computably enumerable-closed: we can enumerate the rational balls intersecting it (\mylabelbt{Proposition 5.16}).

Furthermore, as it can be checked in all the above references of theorems above from \cite{Wei00}, (see also \mylabelwei{Theorem 6.2.1} for the required iteration)  the above reasoning is even effective: we can even produce effectively in $T$  a name of it (even effectively from a name of $\tu f$). This means consequently by doing things in parallel (i.e. dovetailing) that we can effectively enumerate  the rational balls intersecting $\closure{\bigcup_{T}\REACH{\mP,T}(.,.)}$, by considering increasing $T$ and the balls in these enumerations.  
%This is also then true for $\closure{\bigcup_{T}\REACH{\mP,T}(.,.)}$.
%
\end{proof}

A closed set is called co-c.e. closed if one  can effectively enumerate the rational closed balls in its complement. Using arguments similar to the proof of Theorems \ref{pRcorec} and \ref{thcinq}:

%Now:

\begin{theoremd}	\label{thcoreanalyserecursive}
Consider a  computable locally Lipschitz  discrete time  system whose domain $X$ is a computable compact.

For all  $\vx$, $\closure{\REACHP_{\omega}(\vx)} \subseteq \R^{d}$ is a co-c.e. closed subset.
\end{theoremd}

\begin{proofmanon}
From the proof of Theorem \ref{pRcorec}, for all the $\vx$, $\vy$ such that $\REACHPepsilon(\vx,\vy)$ is false, in a easily controllable computable neighborhood of $\vx$, with $\epsilon=2^{-n}$,  there exists some $\delta=2^{-m}$ and some witness $\oR=\oR_{\delta}$ at level $\delta$ of that fact: this witness guarantees $\REACHP(\vx) \subseteq \REACHPepsilon(\vx) \subseteq \oR_{\delta}$, and $\cB(\vy,\eta) \cap \REACHPepsilon(\vx) = \emptyset$ implies $\cB(\vy,\eta) \cap \oR_{\delta} = \emptyset$.
. 

%If we call $d=\distance{\vy}{\closure{\oR_{\delta}}}$, then for any $d' < d$, we are sure that the closure of the ball centered in $\vy$ of radius $d'$ is  not intersecting $\REACHP_{\omega}(\vx)$, as $\distance{\vy}{\closure{\REACHP(\vx)}} \ge \distance{\vy}{\closure{\REACHP_{\delta}(\vx)}} \ge \distance{\vy}{\closure{\oR_{\delta}}}= d > d'$. 

Then a strategy to produce all the rational balls whose closure is not intersecting $\closure{\REACHP(\vx)}$, for  increasing $n$, generate in parallel all such balls in the corresponding witness. This will exhaust all such balls. 
%
%
%test in turn all the rational $\vy$ and witnesses $\oR$ of the fact that $\REACHPepsilon(\vx,\vy)$ is false. As soon as one is found, then generate the rational balls as in the previous paragraph for rational $d'$. 
%
%This will exhaust such balls: indeed, if a rational ball centered in $\vy$ of rational radius $d'$ is not intersecting $\closure{\REACHP(\vx)}$, then considering the point $\vy'$ that realizes the distance $\distance{\vy}{\closure{\REACHP(\vx)}}$, we have $\REACHPepsilon(\vx,\vy')$ false, and the above reasoning for that $\vx$, $\vy'$ shows that $d'<d=\distance{\vy}{\closure{\REACHP(\vx)}}$ and hence
%for some witness, we will generate this ball when considering $d'<d$. 
\end{proofmanon}

\begin{corollary} [Robust $\Rightarrow$ computable)]  \label{coro8}
Assume the hypotheses of Theorem \ref{thcoreanalyserecursive}. 

If $\REACHPomega=\REACHP$ then $\closure{\REACHP} \subseteq \R^{d} \times \R^{d}$ is computable.
\end{corollary}

\begin{proof}
This follows from the fact that a closed set is computable iff it is c.e. closed and co-c.e. closed (\mylabelbt{Proposition 5.16}), observing that  above statements are effective given a name of $\vx$.
\end{proof}

For closed sets, the notion of computability can be also interpreted as the possibility of being plotted with arbitrarily chosen precision: here the intuition is that $\vz/2^{n}$ corresponds to some pixel at precision $2^{n}$, and that $1$ is black (i.e. the pixel is plotted black), $0$ is white (i.e. the pixel is plotted white). 
\olivierpasimportant{Enlevé pour gain de plac: Clearly, we would assume a drawing of $A$ to satisfy the following conditions on $f$.}

\begin{theorem}[\mylabelbt{Proposition 5.7},\mylabelwei{pages 127--128}]  \label{thdessin}
For a closed set $A \subseteq \mathbb{R}^k$,  A is computable iff it can be plotted: there exists a computable function $ f: \mathbb{N} \times \mathbb{Z}^k \rightarrow \mathbb{N}$ with $range( f) \subseteq\{0,1\}$ and such that for all $n \in \mathbb{N}$ and $\vz \in \mathbb{Z}^k$
$$
f(n, \vz)= \begin{cases}1 & \text { if } B(\frac{\vz}{2^n},2^{-n}) \cap A \neq \emptyset,
\\ 0 & \text { if } B(\frac{\vz}{2^n},2.2^{-n}) \cap A = \emptyset, 
\\ 0 \text { or } 1 & \text { otherwise. }\end{cases}
$$
\end{theorem}
This is called \motnouv{local computability} in \cite{MscBraverman}. 

\begin{corollaryd} [Robust $\Rightarrow$ drawable)]  \label{coroplot}
Assume the hypotheses of Theorem \ref{thcoreanalyserecursive}. 

If $\REACHPomega=\REACHP$ then $\closure{\REACHP} \subseteq \R^{d} \times \R^{d}$  can be plotted.
\end{corollaryd}

\begin{proofmanon}
This follows from Corollary \ref{coro8} and Theorem \ref{thdessin} (that is to say \mylabelbt{Proposition 5.7},\mylabelwei{page 127--128}).
\end{proofmanon}

This is even effective in a name of $\tu f$. 
Actually, the converse is true, if some topological  properties are assumed.

\begin{theorem}
Assume ${\REACHP}$ is closed, and can be plotted effectively in a name of $\tu f$. %, and assume that for all $\vx$, $\REACHP(\vx)$ is closed.
Then the system is robust, i.e. $\REACHPomega=\REACHP$. 
\end{theorem}

Actually,  we prove the stronger statement that, if $\closure{\REACHP}$ can be plotted effectively in a name of $\tu f$, then $\REACHPomega(\vx,\vy)=\REACHP(\vx,\vy)$ except maybe for some $(\vx,\vy)$ in $\closure{\REACHP}-\REACHP$.

%
%Assume $\REACHP$ is 
%
%Assume that there is no $\tu y \in \closure{\REACHP}-\REACHP$ that is not not reachable in finite time from $\vx$. \end{theorem}
%
%\olivier{Suis-je convaincu: Je crois. Avis?}

\begin{proof}
By Theorem \ref{thdessin}, we know that  $\closure{\REACHP}$ is computable, and it is known to be equivalent to the fact that the distance function $\distance{\cdot}{\closure{\REACHP}}$ is computable \mylabelwei{Corollary 5.1.8}.
 That means that given some rational ball, a name for $\vx$, and for $\vy$, with $\neg \REACHP(\vx,\vy)$, the following procedure is guaranteed to terminate, when $(\vx,\vy)$ is not in $\closure{\REACHP}-\REACHP$: compute a name of $\distance{(\vx,\vy)}{\closure{\REACHP(\vx)}}$ until a proof that it is strictly positive is found: 
 $\distance{(\vx,\vy)}{\closure{\REACHP(\vx)}}=0$ would mean that $(\vx,\vy) \in \closure{\REACHP}$, but not in $\REACHP$.

 It answers by reading a finite part, say $m$ cells, of the names of $\vx$, $\vy$ and $\tu f$, and hence give the same answer if the names are altered after  symbol number $m$. That means there exists some precision $\epsilon$ (related to $m$, basically $2^{-m}$ if we consider names converging exponentially fast)  so that  $\neg \REACHP(\vx,\vy)$ remains true for some $\epsilon$-neighborhood of $\vx$ and $\vy$, and unchanged by a small variation of $\tu f$. In other words, for all $\vx$, $\vy$, when
$\neg \REACHP(\vx,\vy)$, 
there exists some $\epsilon$ such that $\neg \REACHPepsilon(\vx,\vy)$, i.e. $\neg \REACHPomega(\vx,\vy)$. 
When $\REACHP(\vx,\vy)$ holds, we always have that $\REACHPomega(\vx,\vy)$  holds.
\end{proof}

%\begin{remark}
%Whenhe reachable sets are closed. 
%The previous argument show that robustness is a concept very close to the model of computability in CA, at least when assuming that 
%\end{remark}

This is also possible to adapt this at the complexity level with hypotheses in the spirit of previous results. This would lead to a concept of \motnouv{local poly-space computability} 
%(basically function $f$ is computable in a space polynomial in $n$ and the length of other arguments 
in the spirit of the local poly-time complexity introduced in \cite{MscBraverman}. The latter is devoted to discussing equivalence at the poly-time complexity of various representations of compact sets.

%\olivier{Il y a t'il une réciproque? Liée aux trucs d'avant en fait. Je peux m'en occuper.}

%\olivier{Complexity et pas seulement calculabilité: formuler. Ou pas, d'ailleurs.}

\section{Continuous time and Hybrid  Systems}
\label{sec:continuoustime}

The previous approaches have a very high level of applicability, and are able to talk about systems that could be even continuous time, or hybrid.

\olivierplan{subsection{Basic definitions}}

\begin{definition}
A \motnouv{continuous-time dynamical system} $\mP$ is given by a
set $X \subseteq \mathbb{R}^d$, and some Ordinary Differential Equation (ODE) of the form $\dot{\mathbf{x}}=\tu f(\mathbf{x})$ on $X$.
\end{definition}

It is known that the maximal interval of existence of solutions can be non-computable, even for computable ODEs \cite{dsg06a}. To simplify the discussion, we assume in this article that the ODEs have solutions\footnote{Notice that a non-total solution must necessarily leave any compact, see e.g. \cite{Hartman}, so when $X$ is compact this is not  a restriction.} defined over all $\R$. A trajectory of $\mP$ starting at some $\mathbf{x}_0 \in X$ is a solution of the differential equation with initial condition $\dot{\vx}=\vx_0$, defined as a continuous right derivable function $\tu \xi: \mathbb{R}^{+} \rightarrow X$ such that $\tu \xi(0)=\tu f(\tu x_0)$ and for every $t$, $\tu f(\tu \xi(t))$ is equal to the right derivative of $\tu \xi(t)$.
To each continuous time dynamical system  $\mP$ we associate its reachability relation $\REACHP$ as before. 
%(\cdot, \cdot, \cdot)$ on $\Q^d \times \Q^{d} \times \N$. Namely, for two rational points $\vx$, $\vy$, $\eta=1/2^{n}$ encoded by $n$, the relation $\REACHP(\vx, \vy,\eta)$ holds iff there exists a trajectory of $\mP$ from $\vx$ to $B(\vy,\eta)$.

\olivierplan{subsection{Perturbed continuous time dynamical system}}

%Consider a continuous time dynamical system $\mP$ described by an ODE $\dot{\mathbf{x}}=f(\mathbf{x})$. 

For any $\varepsilon>0$ the $\varepsilon$-perturbed system $\mP_{\varepsilon}$ is described by the differential inclusion $\distance{\dot{\vx}}{\tu f(\vx)}<\varepsilon$. This non-deterministic system can be considered as $\mP$ submitted to a small noise of magnitude $\varepsilon$. We denote reachability in the system $\mP_{\varepsilon}$ by $\REACHPepsilon$. The limit reachability relation $\REACHPomega$ is introduced as before.

%\olivier{Formuler exactement comme pour les systèmes dynamiques à temps discret}
%\manon{c'est normal qu'il soit encore là ce commentaire?}

\olivierpasimportant{Partie pour inspiration: 

\begin{definition}[{\cite[Definition 2.6]{dsg06a}}]
Let $E=\bigcup_{n=0}^{\infty} B\left(a_n, r_n\right) \subseteq \mathbb{R}^m$ be a recursively enumerable open set, where $a_n \in \mathbb{Q}^m$ and $r_n \in \mathbb{Q}$ yield computable sequences satisfying $\overline{B\left(a_n, r_n\right)} \subseteq E$. A function $f: E \rightarrow \mathbb{R}^m$ is called effectively locally Lipschitz in the second argument if there exists a computable sequence $\left\{K_n\right\}$ of positive integers such that
$$
|f(t, x)-f(t, y)| \leq K_n|y-x| \text { whenever }(t, x),(t, y) \in \overline{B\left(a_n, r_n\right)} \text {. }
$$
\end{definition}

\begin{theorem}[{\cite[Theorem 3.1]{dsg06a}}]
 Let $E \subseteq \mathbb{R}^{m+1}$ be a recursively enumerable open set and $f: E \rightarrow$ $\mathbb{R}^m$ be an effectively locally Lipschitz function in the second argument. Let $(\alpha, \beta)$ be the maximal interval of existence of the solution $x(t)$ of the initial-value problem $\dot{x}=f(t, x) ; \quad x(0)=x_0$, where $\left(t_0, x_0\right)$ is a computable point in $E$. Then:
 \begin{enumerate}
 \item  The operator $\left(f, x_0\right) \mapsto(\alpha, \beta)$ is semicomputable (i.e. $\alpha$ can be computed from above and $\beta$ can be computed from below), and
\item The operator $\left(f, x_0\right) \mapsto x(\cdot)$ is computable.
\end{enumerate}
\end{theorem}

This has been improved to weaker hypotheses in \cite{collins2009effective}.

}

\begin{theorem} [Perturbed reachability is co-r.e.]  \label{pppRcorec} 
Consider a  continuous time dynamical system, with $\tu f$ locally Lipschitz and computable, and whose domain is a computable compact.
%\olivier{Sur quel ensemble?}

Then, for all $\vx$, $\closure{\REACHP_{\omega}(\vx)} \subseteq \R^{d}$ is a co-c.e. closed subset.
\end{theorem}

Its proof can  be considered as the main technical result established in \cite{PVB95}.  Independently:

\begin{proof}
The proof is similar to the proof of Theorems \ref{pRcorec} and \ref{thcinq}: adapt the construction of the involved graph $G_{m}$ to cover the flow of the trajectory. With our hypotheses the solutions are defined over all $\R$. It is proved in \cite[Theorem 3.1]{dsg06a} that Lipschitz (and even effectively locally Lipschitz) homogeneous computable ODEs have computable solutions over their maximal domain, so this is feasible.
\end{proof}

\begin{corollary}[Robust $\Rightarrow$ decidable]   
Assume the hypotheses of Theorem \ref{pppRcorec}. 
% If $\REACHPomega=\REACHP$ then $\REACHP$ is decidable.
If $\REACHPomega=\REACHP$ then $\closure{\REACHP} \subseteq \R^{d} \times \R^{d}$ is computable.
\end{corollary}

We can also state the equivalent of previous statements at the complexity level, assuming basic hypotheses that make computability of the solutions to remain in polynomial space. 
\olivier{commenter}

Actually, we can even deal with the so-called hybrid systems. Various models have been considered in literature, but one common point is that they all correspond to continuous time dynamical systems, where the dynamics might be discontinuous (hence it is not computable). 
 To be very general, a dynamical system can be described by its flow $\phi(\vx,t)$ (the idea is that, given $\vx$, and time $t$, $\phi$ maps the trajectory starting from $\vx$ to the position at time $t$). By considering $T$ to be discrete, this even covers discrete time dynamical systems, and $T=[0,+\infty)$,  continuous and hybrid time systems.

\begin{definition}
A \motnouv{hybrid system} $\mP$ is given by a
set $X \subseteq \mathbb{R}^d$, and a semi-group $T$, and some flow function $\phi: X \times T \to X$ satisfying $\phi(\vx,0)=\vx$ and
$\phi(\phi(\vx,t),t')=\phi(\vx,t+t')$.
\end{definition}

 Previous proofs are basically using  the fact that  
 \begin{enumerate}
 \item  reachability $\REACHP$ is c.e;
 \item perturbed reachability is co-c.e.
\end{enumerate}
The first point is usually very clear in any of the considered models, as it is, roughly speaking, always expected that one can at least simulate the model. %This could be assumes as an hypothesis in above definition. 
The second point is, on a given class of models, usually  less clear, but if we look at our proof methods, typically the proof of Theorems \ref{pRcorec} and \ref{thcinq}, we see that we only need to be able to construct some computable abstraction satisfying Claim 1 and Claim 2.  
%Basically, Claim 1 is about the possibility of over estimating (covering by some open balls) the states reachable from a given ball. Claim 2 is about the possibility of under estimating them. %, in a controlled way.  

The key remark is that these properties are not talking about function $\tu f$, but its graph. 
\olivierplusimportant{Enlevé gain de place: basically, Claim 1 says that we are able to cover $\closure{\REACHP(V)}$ effectively for a neighborhood of $\vx$ by rational open balls. Claim 2 says that, given $\epsilon=2^{-m}$, we can find some $\delta=2^{-n}$ so that not reachable vertices of the graph would have their closure not intersecting $\closure{\REACHP(V)}$. This is a property of the graph of the function, i.e. of the flow $\phi$ of the system.
}
The major point is that assuming a function has the closure of its graph computable is a way more general concept than assuming computability. For example, the characteristic function $\chi_{[0,\infty)}$ is not computable, as it is not continuous. But its graph, as well as its closure, is easy to draw (made of two segments): see discussions for e.g. in \cite{MscBraverman}. One usually expects to be able to draw the closure of the graph of the flow, and this is sufficient to get results similar to the previous ones, relating robustness to computability.  
In particular, this allowed us to talk about discontinuous functions, in particular not computable in CA, as we did.

\olivierplusimportant{Enlevé pour gain de place:
This leads to various generalizations, extending the case of Piecewise Constant Derivative systems (that may be considered as of PAM  for continuous times) considered in \cite{asarin01perturbed}. 
By lack of space, we prefer not formalizing this completely, but rather present some other interesting remarks, and relations to other concepts. 
}

\olivierpasimportant{
Pour l'instant manque preuve:

\begin{theoremaprouver} \label{thdirectionunpbisbis}  \label{commeavant}
Assume that the domain $X$ of $\tu f$ is a closed rational box.
\olivier{les hypothèses qui vont bien sur $\tu f$.} 

Then $\REACHPERTURBEDSPACE{\P}{p}\in \PSPACE$
over $\Q^{d} \times \Q^{d} \times \N$.
\end{theoremaprouver}

\begin{proof}

\end{proof}

\begin{theorem}[Polynomially robust to precision $\Rightarrow$ $\PSPACE$]
Assume the hypotheses of Theorem \ref{commeavant}. 

If 
$\REACHP = \REACHPERTURBEDSPACE{\P}{p}$ for some polynomial $p$, 
then $\REACHP \in \PSPACE$.
\end{theorem}

\begin{proof}
	We  have $\REACHPERTURBEDSPACE{\P}{p} \in \PSPACE$ by Theorem \ref{thdirectionunpbisbis} and since $\REACHP = \REACHPERTURBEDSPACE{\P}{p}$, then $\REACHP \in \PSPACE$.
\end{proof}
}

\section{Other perturbations}
\label{sec:otherperturb}
%
%Up to this point, we consider space perturbations, and related robustness.
%
%\olivierplan{Ca c'st nous. Là, on parle bien de. perturbations de machines de Turing}
%%\olivier{New}

We can also consider time perturbed TM: the idea, is that given $n>0$, the $n$-perturbed version of the machine $\M$ is unable to remain correct after a time $n$:
%
%%\olivier{Parler de length?: 
%%We could also consider length perturbed Turing machine: the idea, is that given $n>0$, the $n$-perturbed version of the machine $\M$ is unable to remain correct after a length $n$.
%%}
%
%\olivierplusimportant{New: nouvelle vairante +  je change la définition précédente sur length pour que ca soit clair que ca marche.}
  given an integer $n>0$, the $n$-perturbed version of the machine $\M$ is defined exactly as $\M$ except that after a time greater than $n$ 
  %(if the machine has not already halted)
   then its internal state $q$ can change in a non-deterministic manner:
given a configuration
$
\CONFIGMT{q}{\cdots a_{-n-1} a_{-n} a_{-n+1} \cdots a_{-1}}{ a_0 a_1 \cdots a_{n-1} a_n a_{n+1} \cdots}
$
(with $\neg (q \in F \cup R)$) 
the $n$-perturbed version of $\M$ may go to $
\CONFIGMT{q'}{\cdots a_{-n-1} a_{-n} a_{-n+1} \cdots a_{-1}}{ a_0 a_1 \cdots a_{n-1} a_n a_{n+1} \cdots}
$
for any \footnote{In particular, may accept. More subtle perturbations can be considered, keeping the results valid.} $q' \in Q$. 
% Hence, the machine becomes a nondeterministic transition system.

\olivierplusimportant{
Now, let us introduce the concept of length perturbed Turing machines (PTMs for short). Given an integer $n>0$, the $n$-perturbed version of the machine $\M$ is defined exactly as the time perturbed machine, replacing time by length. 
}

\olivierplusimportant{Some text and results stolen from \cite{asarin01perturbed}, puis variatne sur length}
%\olivier{Variante}

Let $\PERTURBEDT{\M}{n}$ be the time $n$-perturbed language of $\M$, i.e., the set of words in $\Sigma^*$ that are accepted by the time $n$-perturbed version of $\M$. From definitions, and using similar ideas:

%\begin{lemmamanon}
%[{\cite[Lemma $1$]{asarin01perturbed}}]
%$$L(\M) \subseteq \cdots \subseteq \PERTURBEDT{\M}{2}  \subseteq \PERTURBEDT{\M}{1}$$
%\end{lemmamanon}
%

\begin{lemma}
	$L(\M) \subseteq  L^\omega(\M)  \subseteq \cdots \subseteq \PERTURBEDT{\M}{2}  \subseteq \PERTURBEDT{\M}{1}.$
\end{lemma}

%\olivier{Et sa version avec exposant, $\overline{ \text{qui n'est pas écrit.} }$ }

%\olivier{variante}

\begin{theoremd} \label{thpreuveun}
 $L^{\omega}(\M)$ is in the class $\Pi_1^0$.
\end{theoremd}

\begin{proofmanon} \label{preuveun}
	For a word $w$, $w\not\in L^\omega(\M)$, iff   there exists $n\in\N$ such that $w\not\in L^n(\M)$. As $L^{n}(\M)$ is decidable uniformly in $n$,  the complement of $L^\omega(\M)$ is computably enumerable, as it is the uniform in $n$ union of decidable sets.  We get that $L^{\omega}(\M) \in \Pi_1^0$ (co-computably enumerable).
\end{proofmanon}

%\olivier{variante}

\begin{corollarymanon}[Length robust $\Rightarrow$ decidable]
 If $L^{\omega}(\M)=L(\M)$ then $L(\M)$ is decidable.
\end{corollarymanon}

%\begin{corollary}[Length robust $\Rightarrow$ decidable]
% If $L^{\omega}(\M)=L(\M)$ then $L(\M)$ is recursive.
%\end{corollary}

%\olivier{Some text and results stolen from \cite{asarin01perturbed}}

%\olivier{variante}

\begin{theoremd} \label{thpreuvedeux} 
When $M$ always stops, $L^{\omega}(\M)=L(\M)$.
\end{theoremd}

\begin{proofmanon}
	% \olivierpasimportant{.}
	We directly have $L(\M) \subseteq L^\omega(\M)$.
	
	Let $w\in L^\omega(\M) = \bigcap_{n\in\N} L^n(\M)$, so $\forall n \in \N$, $w\in L^n(\M)$. 
	By contradiction, we assume that $L(\M) \subsetneq L^\omega(\M)$. So there exists $w \in L^\omega(\M)$ such that $w \not\in L(\M)$. Since $\M$ always terminates, it rejects $w$  after using a time $q(\length(w))$.
	But, then $w\not\in L^{n}(\M)$ for any $n \ge q(\length(w))+2$, and hence $w\not\in L^\omega(\M)$, a contradiction.
\end{proofmanon}

\olivierplan{Some text and results stolen from \cite{asarin01perturbed}}

\olivierpasimportant{
\begin{theorem} \label{thpreuvetrois}
	For every TM $\M$, we can effectively construct another $T M \M^{\prime}$ such that $L^\omega\left(\M^{\prime}\right)=\overline{L(\M)}$
\end{theorem}
}

\olivierplan{section{New results: TIME (for Turing machines)}}

\begin{definition}
Given some function $f: \N \to \N$, we write $\PERTURBEDTIME{\M}{f}$ for the set of words accepted by $\M$ with time perturbation $f$: 
{	$\PERTURBEDTIME{\M}{f} = \{w | w\in \PERTURBEDT{\M}{f(\length(w))}\}.$}
\end{definition}

\olivierpasimportant{ Enlevé pour gain de place:
\begin{theorem} \label{thdirectionundd}
$\PERTURBEDT{\M}{n} \in DTIME(p(n))$
for some polynomial $p$.
\end{theorem}
}

\begin{theoremd}[Polynomially  robust to time $\Leftrightarrow$ $\PTIME$] \label{mainptimeone}
A language $L$ is in $\PTIME$ iff  for some $\M$ and some polynomial $p$, 
$L= L(\M)=\PERTURBEDTIME{\M}{p}.$
\end{theoremd}

\begin{proof}
This can be established as for space perturbation. In an independent view, the intuition of the proof is that the polynomial in $n$ can be seen as a time-out. $M$ works in polynomial time $p(n)$, so in at most $p(n)$ steps, so the machine for $L^n(M)$ can reject if it has not accepted or rejected in $p(n)$ steps.
	
	($\Rightarrow$) If $M$ always terminates and works in polynomial time, then there exists a polynomial $q$ that bounds the execution time of $M$, so we have a polynomial $p$  ($p\geq q$) such that, for $n\in N$, $\PERTURBEDT{\M}{p}  \subseteq L(M)$. We have the other inclusion by definition.
	
	($\Leftarrow$) We always have $\PERTURBEDT{\M}{p(n)} \in \PTIME$ \olivierpasimportant{arnaque de pas dire que? by Theorem \eqref{thdirectionundd}}
	 and since $\PERTURBEDT{\M}{p(n)} = L$, then $L \in \PTIME$.
\end{proof}

\begin{theorem}[Polynomially robust to length $\Leftrightarrow$ $\PTIME$] 
Any  $\PTIME$ language is reducible to the reachability relation of PAM with $\REACHP = \PERTURBEDTIME{\P}{p}$ for some polynomial $p$.
\end{theorem}

\newcommand\ddistance[2]{\delta(#1,#2)}

Fix some distance $\ddistance{\cdot}{\cdot}$ over the domain $X$. A finite trajectory of discrete time dynamical system $\mP$ is a finite sequence $(x_t)_{t \in 0\dots T}$ such that ${x}_{t+1}= f\left({x}_t\right)$ for all $0 \le t <T$. Its  associated \motnouv{length} is defined as $\mL = \sum_{i=0}^{T-1} \ddistance{x_{i}}{x_{i+1}}$.

We could also consider length perturbed discrete time dynamical system: the idea, is that given $L>0$, the $L$-perturbed version of the system is unable to remain correct after a length $L$.
We then define $\REACHPERTURBEDL{\mP}{L}(\vx,\vy)$ as there exists a finite trajectory of $\mP$ from $\vx$ to $\vy$ of length $\mL  \le L$.
\olivier{Truc louche. Pas inclusions entre classes comme avant. Grave docteur?}
%
%Given some function $f: \N \to \N$, we write $\REACHPERTURBEDLENGTH{\mP}{f}$ is the set of words accepted with length perturbation $f$: 
%{	$\REACHPERTURBEDLENGTH{\M}{f} = \{w | w\in \REACHPERTURBEDL{\M}{f(\length(w))}\}.$}

When considering TMs as such dynamical systems, $\ddistance{\cdot}{\cdot}$ is basically some distance over configurations of TMs.  Word $w$ is said to be accepted in length $d$ if the trajectory starting from $C_{0}[w]$ to the accepting configuration has length $\le d$.

%\olivierpourmanon{meilleur nom que time-metric?}
%
%\olivier{Semble sous-optimal. Pourquoi faire dépendre de C'?}
\begin{definition} \label{deftimemetric}
Distance $\ddistance{C}{C'}$  is called \emph{time metric} iff for $C \vdash C'$, we have $\ddistance{C}{C'} \le p(\length(C))$, 
%+ \length(C')$), 
and $\ddistance{C}{C'} \geq \frac{1}{p(\length(C))}$ % + \length(C'))}$, 
for some polynomial $p$.
\end{definition}

Write $\LENGTH{M}{t}$ for the set of words accepted by $M$ in a length less than $t$.
%
%\begin{definition}
Given some function $f: \N \to \N$, we write $\EXACTLENGTH{\M}{f}$  for 
{$\EXACTLENGTH{\M}{f} = \{w | w\in \LENGTH{\M}{f(\length(w)}\}.$}
%\end{definition}

\begin{theoremd}[Length robust  for some time-metric distance $\Leftrightarrow \PTIME$] \label{thmSympaTime}
Assume $\ddistance{\cdot}{\cdot}$ is time metric.
Then, 
a language $L$ is in $PTIME$ iff for some Turing machine $\M$, and some polynomial $p(n)$, 
$L=L(\M)=\EXACTLENGTH{\M}{f}.$
\end{theoremd}

\begin{proof}
Let $w$ be the input of size $n$. The execution of a Turing 
	machine is a sequence $(C_i)= \CONFIGMT{q_{i}}{l_i}{ r_i}$.
	
	$(\Rightarrow)$ If $L$ is in $\PTIME$, so there is a Turing machine 
	$\M$ that computes $L$ in polynomial time $p(n)$. Since the 
	distance between two successive configurations is bounded by a polynomial 
	$q(n)$, we have that the total length $\mL$ is $\mL=\sum_{i=0}^{p(n)-1} \distance{C_i}{C_{i+1}} \le \sum_{i=0}^{p(n)-1} q(n)$, which is a polynomial in $n$. 
	
	Thus $L$ is computable in polynomial length.
	
	$(\Leftarrow)$ If $L$ is computable in polynomial length $p(n)$. 
	
	Let $T$ to be fixed. 
	Then we have, for all $i\in \{ 1 \dots T \}$: $\distance{C_i}{ C_{i+1}}\geq \frac{1}{poly(\length(C_{i}))}$, thus $\sum_{i=0}^{T-1} \distance{C_i}{ C_{i+1}}\geq \frac{T}{poly(\length(C_{min}))} $, where $C_{min}$ is chosen to minimize the previous lower bound. 
	
	Take $T = p(n) \times poly(\length(C_{min}))$, we have that a Turing machine simulating the trajectory will accepts or rejects after a polynomial number of steps, thus $L \in \PTIME$.	
%\olivier{besoin que $|C|+|C'|$ reste minoré}
%	
%	We denote by $T$ the number of steps necessary.
%	Then we have, for all $i\in \{ 1 \dots T \}$: $ d(C_i, C_{i+1})\geq \frac{1}{poly(|C| + |C'|)}$, thus $\sum_{i=0}^{T-1} d(C_i, C_{i+1}) \geq \frac{T}{poly(|C| + |C'|)} $. So, with $T = p(n) \times poly(|C| + |C'|)$, we have that the machine accepts or rejects after a polynomial number of steps, thus $L \in PTIME$.	
\end{proof}

%\begin{proofmanon}
%	\underline{Intuition:} the polynomial in $n$ can be seen as a time-out. $M$ works in polynomial time $p(n)$, so in at most $p(n)$ steps, so the machine for $L^n(M)$ can reject if it has not accepted or rejected in $p(n)$ steps.
%	
%	\begin{itemize}
%		\item $(\Rightarrow)$ If $L \in PTIME$, then there exists a Turing machine $\M$ that computes $L$ in polynomial time $q$, so we have a polynomial $p \geq q$, such that $L^{\{p\}}(\M) \subseteq L(\M) $. We have the other inclusion by definition.
%		\item $(\Leftarrow)$ If we have a Turing machine $\M$, and a polynomial $p$ such that $ L = L(\M)=\PERTURBEDTIME{\M}{p}$, then $L$ can be computed in $p(n)$ steps, with $n$ the size of the input, thus $L$ is computable in polynomial time.
%	\end{itemize}
%	
%\end{proofmanon}

One way to obtain such a distance $\ddistance{C}{C'}$ is to take the Euclidean between $\gammac(C)$ and $\gammac(C')$ %(see Section \ref{sec:simulation}) 
for $\gamma=\gammacompact$. 

\begin{propositiond} The obtained distance is time metric. \label{camarche}
\end{propositiond}

\begin{proof} %\textcolor{red}{VALIDÉE}
%\olivier{suis-je d'accord que pas majoration très très brutale.}
%\olivier{Besoin de dire des trucs spécifiques sur codage, et vérifier que bien bon codage.}
	We consider $C_i = \CONFIGMT{q_i}{ {l_i}}{ {r_i}}$ and $C_{i+1}=\CONFIGMT{q_{i+1}}{{l_{i+1}}}{{r_{i+1}}}$, with $C_{i} \vdash C_{i+1}$. We write $\overline{l_{i}}= \gammac(l_{i})$ and $\overline{r_{i}}= \gammac(r_{i})$.
	
	Then, from the definition of $\gammac$ and $\gammacompact$:  

	\begin{itemize}
		\item $\left| \overline{r_{i+1}}  - \overline{r_{i}} \right| \le 1$.
		\item  $\left|  \overline{l_{i+1}} - \overline{l_{i}}  \right| \leq 1$.
%		\olivier{dans la phrase d'après, on dit que c'est au plus 1, (et c'est l'item précédent non), et du coup, on comprend pas à quoi ca sert de dire le texte qu'il y a avant. Enlever? le premier item suffit?}
		\item And the gaps between  $\overline{r_{i+1}}$ and $\overline{r_{i}}$, and $\overline{l_{i+1}} $ and $\overline{l_{i}}$ remain polynomial in  the size of a configuration.	\end{itemize}
  This provides property 1.
  
  By the encoding of the real numbers over the tapes of the Turing machines, the gap between two consecutive configurations is at least $\frac{1}{2} $
  (we assume that the Turing machine is not allowed not to do anything: that would clearly corresponds to a looping situation).
  This provides property 2.
	
\end{proof}

Given some function $f: \N \to \N$, we write $\REACHPERTURBEDLENGTH{\mP}{f}$ for the set of words accepted by $\M$ with length perturbation $f$: \\
	$\REACHPERTURBEDLENGTH{\M}{f} = \{w | w\in \REACHPERTURBEDL{\M}{f(\length(w)}\}.$

\begin{theoremd}[Polynomially length robust $\Leftrightarrow \PTIME$] \label{encoreun}
Assume distance $d$ is time metric. 
Assume $\REACHP= \REACHPERTURBEDLENGTH{\mP}{p}$ for some polynomial $p$.
Then $\REACHP$ is in $\PTIME$. 
\end{theoremd}

\begin{proof}

	Since $d$ is time metric, a polynomial time and polynomial length are essentially the same, so the proof is very analogous to the one of  Theorem \ref{thmSympaTime}. 

\end{proof}

\section{Analog complexity under robustness Prism}
\label{sec:jebrode}

%\olivier{Définir length pour une trajectoire d'un PPAM, d'un système dynamique à temps continu: ok, ca je vais tenter de m'en occuper}
%
%\olivier{Et qu'est ce qu'on obtien pour le cas temps continu?: ok, ca je vais tenter de m'en occuper}

\newcommand\FPTIME{\cp{FPTIME}}
\newcommand{\Rp}{\R_{+}}
\newcommand{\glen}[1]{\myop{len}_{#1}}
\newcommand\psik{\gammacompactk}
\newcommand\myOmega{\amalg}

The following was established in \cite{JournalACM2017vantardise} ($\glen{\tu y}(0,t)$ stands for the length of the curve between time $0$ and $t$):

\begin{theorem}[{Analog characterization of $\PTIME$ \cite[Theorem 2.2]{JournalACM2017vantardise}}]\label{th:p_gpac_un}
A decision problem (language) $\mathcal{L}$  belongs to the class $\PTIME$
if and only if it is poly-length-analog-recognizable. That is to say: there exist vectors $\tu p$, and $\tu q$ of polynomials with rational coefficients 
%with $d$ variables, both with coefficients in $\Q$, % (resp. in $\K$)
and a polynomial $\myOmega: \Rp \to \Rp$,
such that for all $w\in\Sigma^*$,
there is a (unique) $\tu y:\Rp\rightarrow\R^d$ such that for all $t\in\Rp$: \\
1)  $\tu y(0)=\tu q(\psik(w))$ and $\tu y'(t)=\tu p(\tu y(t))$ \\
%\hfill$\blacktriangleright$ $y$ satisfies a differential equation
2)  if $|\tu y_1(t)|\geqslant1$ then 
$|\tu y_1(u)|\geqslant1$ for all $u\geqslant t$\\
%\hfill$\blacktriangleright$ decision is stable
3) if $w\in\mathcal{L}$ (resp. $\notin\mathcal{L}$) and $\glen{\tu y}(0,t)\geqslant\myOmega(|w|)$
then $\tu y_1(t)\geqslant1$ (resp. $\leqslant-1$) \\
%\hfill$\blacktriangleright$ decision
4)  $\glen{\tu y}(0,t)\geqslant t$
%\hfill$\blacktriangleright$ technical condition
%%\footnote{This could be
%  replaced by only assuming that we have somewhere the additional
%  ordinary differential equation $y'_0=1$.}
%\end{enumerate}
\end{theorem}

Conditions 2) is basically stating that this is a length-robust system. Condition 4) is the equivalent of second condition of Definition \ref{deftimemetric}, that guarantees the equivalence with length for Turing machines. And the result is obtained using a step-by-step emulation using encoding $\psik$. 
\newcommand{\myclass}[1]{\operatorname{#1}}
\newcommand{\mygpc}[1][]{\ensuremath{\myclass{ALP}_{#1}}}
\newcommand\K{\mathbb{K}}
\olivierpasimportant{Commenté: 
\olivierpasimportant{plein de notations trucs pas definis là}
%
%
%\olivier{Et donc, expliquer le rapport avec la choucroute.}

Actually, a key argument is the following:

\begin{definition}[Emulation]\label{def:fp_gpac_embed_multi}
$f:\left(\Gamma^*\right)^n\rightarrow\left(\Gamma^*\right)^m$ is
called  \emph{emulable}
if there exists $g\in\mygpc$ and $k\in\N$ such that for any word $\vec{w}\in\left(\Gamma^*\right)^n$, $g(\psik(\vec{w}))=\psik(f(\vec{w}))$ where: $$\psik(x_1,\ldots,x_\ell)=\left(\psi(x_1),\ldots,\psi(x_\ell)\right)$$
and $\psik$ is defined as in Definition~\ref{def:fp_gpac_embed}.
\end{definition}

\begin{theorem}[Multidimensional $\FP$ equivalence]\label{th:fp_gpac_multi}
For any % generable field $\K$ such that $\Rgen\subseteq\K\subseteq\Rpoly$
% and 
$f:\left(\Gamma^*\right)^n\rightarrow\left(\Gamma^*\right)^m$,
$f\in\FPTIME$ if and only if $f$ is % \neural{}-
emulable.
\end{theorem}

where
%\olivier{utiliser $\gammacompact$ comme notation.}

\begin{definition}[{Discrete emulation \cite[Definition 6.1]{JournalACM2017vantardise}  }]\label{def:fp_gpac_embed}
$f:\Gamma^*\rightarrow\Gamma^*$ is called $\K$-\emph{emulable} if 
there exists $g\in\mygpc[\K]$
and $k\geqslant1+\max(\gamma(\Gamma))$ such that for any word $w\in\Gamma^*$,
$g(\psik(w))=\psik(f(w))$ where 
\[\psik(w)=\left(\sum_{i=1}^{|w|}\gamma(w_i)k^{-i},|w|\right).\]
We say that $g$ $\K$-\emph{emulates} $f$ with $k$. When the field $\K$ is
unambiguous, we will simply say that $f$ is emulable.
\end{definition}
}

Very recently, a characterization of $\PSPACE$ has also been obtained: see \cite{TheseRiccardo}, for full details. Basically, of the form.

\begin{theorem}[{\cite[Theorem 3.11]{GozziPSPACE},\cite{TheseRiccardo}}]
\label{Th:PSPACE}A language $L\subseteq\Gamma^{\ast}$ belongs to
$\operatorname{PSPACE}$ iff there are ODEs 
%\begin{equation}
%\left\{
%\begin{array}
%[c]{l}%
$\tu y^{\prime}=\tu p(\tu y,\tu z)$ and
$\tu z^{\prime}=\tu q(\tu y,\tu z),$
%\end{array}
%\right.  \label{Eq:master_aux}%
%\end{equation}
 where $\tu p,\tu q$ are vectors of polynomials,
% that $\hat{\Psi}%
%$-decides it in polynomial space, that is to say:
%\end{theorem}
%
%\begin{definition}[{\cite[Definition 3.8]{GozziPSPACE},\cite{TheseRiccardo}}]
%\label{Def:Space-decidability}Let $L\subseteq\Gamma^{\ast}$ be a language and
%$g:\mathbb{N}\rightarrow\mathbb{N}$ be a function. Then we say that
%$L$ is $(\Psi$-)decidable 
%by an ODE 
% in space $g$, 
%where $p,q$ are functions formed by polynomial components, if
there is  a polynomial $u$,  (vector-valued) functions $\tu r,\tu s\in\operatorname{GPVAL}$, a $\gammaentierk$-bound
$\phi$, and
$\varepsilon>0,$ $\tau\geq\alpha>0$, $\alpha,\tau\in\mathbb{Q}$,
such that, for all $w\in 
\Sigma^{\ast}$, one has that the solution $(\tu y,\tu z)$ of
the ODEs with the initial condition $\tu y(0)=\tu r(x)$
and $\tu z(0)=\tu s(\tu y(0))=\tu s\circ \tu r(0)$, with $\distance{\gammaentierk(w)}{\tu x}<1/4$, satisfies:\\
1. %\textbf{(halting decision is irreversible)} 
If $\bar{t}_{1}>0$ is such that
$\left\vert \tu y_{1}(\bar{t}_{1}) \right\vert \geq1$, then $\left\vert \tu y_{1}(t) \right\vert \geq1$ for all $t\geq \bar{t}_{1}$ and
$\left\vert \tu y_{1}(t) \right\vert \geq3/2$ for all $t\geq \bar{t}_{1}+1$;\\
2. %(\textbf{halting decision is eventually taken with correct output}) 
If $w\in L$ (respectively $\notin L$)
then there is some $\bar{t}_{1}>0$ such that $\tu y_{1}(\bar{t}_{1})\geq1$ (respectively
$\leq-1$); \\
3.  %(\textbf{bounded space}) 
$\left\Vert (\tu y(t),\tu z(t))\right\Vert \leq
\phi\circ u(|w|)$, for all $t\geq0$;\\
4.  %(\textbf{robustness preserves main properties}) 
Suppose that $(\tilde{\tu y},\tilde
{\tu z})$ satisfies the ODE with %initial condition
$\tu y(0)=\tu r(x)$, $\tu z(0)=\tu s(\tu y(0))=\tu s\circ \tu r(\tu x)$, except possibly at time instants $t_{i}$, for $i=1,2,\ldots$, satisfying:
a)  $t_{i}-t_{i-1}\geq\tau$, where $t_0=0$;
b)  $\|\tilde{\tu y}(t_i)-\lim_{t\to t_i^-}\tilde{\tu y}(t)\|\leq\varepsilon$ and $\tilde{\tu z}(t_i)=\tu s(\tilde{\tu y}(t_i))$;
c)  $\lim_{t\to t_i^-}\tilde{\tu z}_1(t)> 1$.
Then conditions 1,2,3,5 hold for $(\tilde{\tu y},\tilde
{\tu z})$;\\ 
5) %(\textbf{robustness is common}) 
For any $b>a\geq0$ such that $\left\vert
b-a\right\vert \geq\tau$, there is an interval $I=[c,d]\subseteq\lbrack a,b]$,
with $\left\vert d-c\right\vert \geq\alpha$, such that $\tu z_{1}(t)\geq3/2$ for
all $t\in I$.
%\end{enumerate}
\end{theorem}

Condition $\distance{\gammaentierk(w)}{x}<1/4$,  4) and 5) basically impose that there exists some abstraction graph. Conditions 3) makes it keep  a polynomial $\logsize$, and conditions 1) and 2) impose to the system to be eventually decisional.  This guarantees $\PSPACE$, and allows a robust emulation of a TM. However, the space used by the ODEs cannot be ``read'' easily (as by a concept such as length in previous theorem).

This is using a rather natural encoding, but the system is not living in a compact. If one wants to remain bounded, an alternative is to use the trick of \cite{BGHRobust10}, based on a change of variable, at the price of a rather adhoc encoding. This is basically based on the following:

%\olivier{Choucroute 2: Mes oeuvres. Je deviens gateux, je réalise en me relisant qu'on avait prouvé des trucs que j'ai reprouvé là. Ceci étant.
%\bigskip
%\bigskip
%}

%\begin{equation}
%y=w_{0}+w_{1}2+\ldots+w_{n}2^{n}.\label{Eq:encoding:integers}%
%\end{equation}
%
%\olivier{attention sens de robuste là } 
\begin{theorem}[{Robust simulation of a TM over $\R^{6}$ \cite{gcb08}}]
\label{Teo:Graca} For any TM $M$, there is an analytic and
computable ODE $y^{\prime}=g_{M}(y)$ defined over $\mathbb{R}^{6}$ which
simulates $M$ using the encoding $\gammaentier$, % (\ref{Eq:encoding:integers}), and
and remains valid for perturbations less than  $\varepsilon\leq1/4.$
\end{theorem}

%Following \cite{BGHRobust10}, it is possible to use this  construction to simulate $M$ on a compact set
%$X=(-1,1)^{6}.$ 
The idea is that if $\phi$ is a solution of $y^{\prime}=\tu g_{M}(y)$ simulating
$M$ on $\mathbb{R}^{6}$, we can consider $\phi_{1}=\frac{2}{\pi}\arctan\phi$ 
%(and hence $\phi=\tan\left(  \frac{\phi_{1}\pi}{2}\right)  $) 
as a corresponding
emulation of $M$ on $(-1,1)^{6}$.  
Then $\phi_{1}$ will be solution of the ODE $\phi_{1}^{\prime} =\tu f_{M}(\phi_{1})\nonumber$ with
$\tu f_{M}(x)=\frac{2}{\pi}\frac{1}{1+\tan^{2}\left(  \frac{x\pi}{2}\right)  }%
\tu g_{M}\left(  \tan\left(  \frac{x\pi}{2}\right)  \right)  $.
Consequently, the continuous time dynamical system given by  $y^{\prime}=\tu f_{M}(y)$ simulates TM $M$ on $X$, if the input word  $w=a_{1} \dots a_{n}$  is 
encoded in $X$ by  $\gamma_{arctan}(w)=\frac{2}{\pi}\arctan(\gammaentier(w))= \frac{2}{\pi}\arctan(w_{0}+w_{1}2+ \dots +w_{n}2^{n})$. 
%This will preserves the space robustness of the TM $M$, and 
Then a computation will remain correct if the states are not perturbed more than
$\arctan(s(w)+\varepsilon)-\arctan(s(w))$, where $s(w)$ is an upper bound of the size of tape on word $w$.
% (hence assuming that the machine halts on $w$, i.e  that this bound exists).

%
%In general%
%\begin{align}
%\phi_{1}^{\prime} &  =(\frac{2}{\pi}\arctan\phi)^{\prime}=\frac{2}{\pi}%
%\frac{1}{1+\phi^{2}}\phi^{\prime}=\frac{2}{\pi}\frac{1}{1+\phi^{2}}g_{M}%
%(\phi)\text{ }\Longrightarrow\label{Eq:Simul}\\
%\phi_{1}^{\prime} &  =\frac{2}{\pi}\frac{1}{1+\phi^{2}}g_{M}(\phi)=\frac
%{2}{\pi}\frac{1}{1+\tan^{2}\left(  \frac{\phi_{1}\pi}{2}\right)  }g_{M}\left(
%\tan\left(  \frac{\phi_{1}\pi}{2}\right)  \right)  =f_{M}(\phi_{1})\nonumber
%\end{align}
%where%
%\[
%f_{M}(x)=\frac{2}{\pi}\frac{1}{1+\tan^{2}\left(  \frac{x\pi}{2}\right)  }%
%g_{M}\left(  \tan\left(  \frac{x\pi}{2}\right)  \right)  .
%\]
%Hence, the system $y^{\prime}=f_{M}(y)$ simulates $M$ on $X$, with input $w$
%coded by $v(w)$, where $v$ is given by (\ref{Eq:encoding:compact}). Moreover,
%robustness among states still exists, and the simulation of $M$ can be carried
%out if the states are not perturbed more than $
%\bar{\varepsilon}=\arctan(m+\varepsilon)-\arctan(m)\label{Eq:epsilon_bar}%
%$, where $m$ is an upper bound on the states of $M$. 

\olivier{Enlevé pour gain de place:
This was used to prove that for any decidable language $L$, there is a robust analytic and computable continuous-time (resp. discrete) system over
$\mathcal{S}=(-1,1)^{6}$ computing $L$ \cite[Theorem 22]{BGHRobust10} (resp. \cite[Theorem 23]{BGHRobust10}).}  However, the used encoding $\gamma_{arctan}$ is rather artificial, and mainly the system is defined over some open bounded domain, not compact. Furthermore, robustness holds for points close to the image by $\gammac$ of configurations, but not for arbitrary points. 

The question of a simpler characterization of $\PSPACE$, over a compact, using a simple encoding remains open. However, we believe that our results help to understand that space-complexity is related to robustness to precision over a compact, or more generally to the $\logsize$ of some abstraction graph over general domains.

Notice that %given some space bound, by embedding a  simulation of a Turing machine with this space bound, in a robust way, 
we can prove that reachability is $\PSPACE$-complete for polynomially robust pODE systems,  %in the spirit of Theorem \ref{recipir}.
% (and space would be measured by the related precision).  
But the point of the previous results (and open question) is to get a \emph{uniform} embedding: given some machine $M$, we would like a same ODE that works for any input size.

%\begin{theorem}[{\cite[Theorem 22]{BGHRobust10}}]
%\end{theorem}
%
%\begin{theorem}[{\cite[Theorem 23]{BGHRobust10}}]
%For any decidable language $L$, there is a robust analytic and computable discrete-time system over
%$\mathcal{S}=(-1,1)^{6}$ computing $L$.
%\end{theorem}
%
%\olivier{discuté compact/non-compact, encodage $\gamma_{arctan}$ vs naturel}

%%%%%%%%%%%%%%%%%%%%%%%%%%%%%%%%%%%%%%%%%%%%%%%%%%%%
%
%
%  						BIBLIO
%
%
%%%%%%%%%%%%%%%%%%%%%%%%%%%%%%%%%%%%%%%%%%%%%%%%%%%%

\newpage
\bibliographystyle{plain}
\bibliography{bournez,perso}

%%%%%%%%%%%%%%%%%%%%%%%%%%%%%%%%%%%%%%%%%%%%%%%%%%%%
%
%
%  						APPENDIX
%
%
%%%%%%%%%%%%%%%%%%%%%%%%%%%%%%%%%%%%%%%%%%%%%%%%%%%%

\newpage
\appendix

\section*{Crash course in computable analysis}

We  recall here the very basics of computable analysis: see  \cite{Wei00} or \cite{brattka2008tutorial}. 
\olivierplan{Un peu de concepts de l'analyse calculable. Autrees refs? Certaines parties à placer ailleurs. Prise de \cite{Wei00} partie introduction.}

Basically, the idea behind classical computability and complexity is to fix some representations of objects (such as graphs, integers, etc, \dots) using finite words  over some finite alphabet, say $\Sigma=\{0,1\}$, and to say that such an object is computable when such a representation can be produced using a Turing machine. The aim of computable analysis is to be able to talk also about objects such as real numbers, functions over the reals, closed subsets, compacts subsets, \dots, which cannot be represented by finite words over $\Sigma$ (a clear reason for it is that such words are countable while the set $\R$ for example is not).  However, they can be represented by some infinite words over $\Sigma$, and the idea is to fix such representations for these various objects, called \emph{names}, with suitable computable properties. In particular, in all the following proposed representations, it can be proved that an object is computable iff 
it has some computable representation.

\begin{remark}
Here the notion of computability involved is the one of Type 2 Turing machines, that is to say computability over possibly infinite words: the idea is that such a machine has some read-only input tape(s), that contain the input(s), which can correspond to either a finite or infinite word(s), a read-write working tape, and one (or several) write only output tape(s). It evolves as a classical Turing machine, the only difference being that we consider it outputs an infinite words when it  writes forever the symbols of that words on its (or one of its) write-only infinite output tape(s): see \cite{Wei00} for details.
\end{remark}

A name for a point $\vx \in \R^{d}$ is a sequence $(I_{n})$ of nested open rational balls  with $I_{n+1} \subseteq I_n$ for all $n \in \mathbb{N}$ and $\{x\}=\bigcap_{n \in \mathbb{N}} I_n$.   Such a name can  be encoded as infinite sequence of symbols. 

We call a real function $f: \subseteq \mathbb{R} \rightarrow \mathbb{R}$ computable, iff some (Type 2 Turing) machine maps any name of any $x \in \operatorname{dom}(f)$ to a name of $f(x)$. For real functions $\tu f: \subseteq \mathbb{R}^n \rightarrow \mathbb{R}$ we consider machines reading $n$ names in parallel. 

It can be proved that a computable function is necessarily continuous. A name for a function $\tu f$ is a list of all pairs of open rational balls $(I, J)$ such that $\tu f(\bar{I}) \subseteq J$. Following above remark, one can prove that a real function is computable iff it has some computable name. 

A name for a closed set $F$ is a sequence $(I_{n})$ of open rational balls such that $\closure{I_{n}} \cap F=\emptyset$ and a  sequence $(J_{n})$ of open rational balls such that $J_{n} \cap F \neq \emptyset$.

Given some closed set $F$, the distance function $d_F: \mathbb{R}^n \rightarrow \mathbb{R}$ is defined by $d_F(x):=\inf _{y \in F} \distance{x}{y}$. Closed subset $F \subseteq \mathbb{R}^n$ is  computable iff its distance function $d_A: \mathbb{R}^n \rightarrow \mathbb{R}$ is  (\mylabelwei{Corollary 5.1.8}). 
A name for a compact $K$ is a name of $F$ as a closed set, and an integer  $L$ such that $K \subset B(0,L)$.

A closed set is called  computably-enumerable closed if one can effectively enumerate the rational open balls intersecting it: $\left\{(q, \varepsilon) \in \mathbb{Q}^n \times \mathbb{Q}_{+} \mid \mathrm{B}(q, \varepsilon) \cap A \neq \emptyset\right\}$ is computably enumerable (\mylabelbt{Definition 5.13},\mylabelwei{Definition 5.1.1}).
A closed set is called co-computably-enumerable closed if one  can effectively enumerate the rational closed balls in its complement: 
the set $\left\{(q, \varepsilon) \in \mathbb{Q}^n \times \mathbb{Q}_{+} \mid \overline{\mathrm{B}}(q, \varepsilon) \subseteq U\right\}$ is computably enumerable
(\mylabelbt{Definition 5.10},\mylabelwei{Definition 5.1.1}).

We need also the concept of polynomial time computable function in computable analysis: see \cite{Ko91}. In short, a quickly converging name of $\vx \in \R^{d}$ is a name  of $\vx$, with $I_{n}$ of radius $<2^{-n}$. 
A function $\tu f: \R^{d} \to \R^{d'}$ is said to be computable in polynomial time, if there is some oracle TM $M$, such that, for all $\vx$, given any fast converging name of $\vx$ as oracle, given $n$, $M$ produces some open rational ball of radius $<2^{-n}$ containing $\tu f(\vx)$,  in a time polynomial in $n$.

%\section*{Proofs}

%\ANNEXEPREUVEDE{Lemma \ref{lemmapath}}{\preuvelemmapath}

%\ANNEXEPREUVEDE{Lemma \ref{lemma2}}{\preuvelemmadeux}

%\ANNEXEPREUVEDE{Theorem \ref{savitch}}{\preuvesavitch}

%\ANNEXEPREUVEDE{Corollary \ref{corosuccint}}{\proofsavi}

%\ANNEXEPREUVEDE{Theorem \ref{thprisedetete}}{\preuveprisedetete}

%\ANNEXEPREUVEDE{Theorem \ref{thmaindeuxp}}{\preuvethmaindeuxp}

%\ANNEXEPREUVEDE{Theorem \ref{thbaba}}{\preuvethbaba}

%\ANNEXEPREUVEDE{Theorem \ref{thcoreanalyserecursive}}{\preuvethcoreanalyserecursive}

%\ANNEXEPREUVEDE{Corollary \ref{coroplot}}{\preuvecoroplot}

%\ANNEXEPREUVEDE{Theorem \ref{mainptimeone}}{\preuvemainptimeone}

%\ANNEXEPREUVEDE{Theorem \ref{thmSympaTime}}{\preuvethmSympaTime}

%\ANNEXEPREUVEDE{Theorem \ref{camarche}}{\preuvecamarche}

%\ANNEXEPREUVEDE{Theorem \ref{encoreun}}{\preuveencoreun}

\newpage

%\subsection{Proof of Theorem \ref{mainptimeone}}

\olivierplan{
subsection{Notations}
}

%\olivier{En fait, il y a des conflits de notations, et ca rend pas clair de quoi on parle, voici une proposition de notation. Mais à discuter avec Jeune Padawan Manon.}
%\manon{Discuté!}

\olivierpasimportant{
Consider $\M$ a Turing machine. Consider some function $f: \N \to \N$.

\begin{itemize}
\item $L(\M)$ is the set of words accepted by $\M$.
\item $\TIME{\M}{n}$ is the  set of words accepted by $\M$ at a time $\le n$.

\item $\EXACTTIME{\M}{f}$ is the set of words accepted by $\M$ in time $f$: 
	$$\EXACTTIME{\M}{f} = \{w | w\in \TIME{\M}{f(|w|}\}.$$

\item $\SPACE{\M}{n}$ is the  set of words accepted by $\M$ using a space $\le n$.
\item $\EXACTSPACE{\M}{f}$ is the set of words accepted by $\M$ in space $f$: 
	$$\EXACTSPACE{\M}{f} = \{w | w\in \SPACE{\M}{f(|w|}\}.$$

\end{itemize}
}

\olivierpasimportant{
subsection{Trivial facts}

From definitions:

We have $$L(\M) = \bigcup_{n \in \N} \TIME{\M}{n}.$$

We have $$PTIME = \bigcup_{k \in \N} \EXACTTIME{\M}{n^{k}}.$$

We have $$PSPACE = \bigcup_{k \in \N} \EXACTSPACE{\M}{n^{k}}.$$
}

\olivierpasimportant{
subsection{Rather trivial}
Rather trivial (playing with definitions)

\begin{theorem}
$L \in PTIME$ iff for some $\M$ and some polynomial $p$
$$L= L(\M)= \EXACTTIME{\M}{p}.$$
\end{theorem}

\begin{proof}
$\Rightarrow$. Consider $M'$ that accepts $L$ in time $p$. Then $L= \EXACTTIME{\M}{p}$.

$\Leftarrow$. Assume $L= \EXACTTIME{\M}{p}$. Then just simulate $\M$ during time $p$ to determine whether $w \in L$ or not. If you prefer, $\EXACTTIME{\M}{p}$ is in $DTIME(poly(p))$, and as $L= L(\M)= \EXACTTIME{\M}{p}$ so is $L$.
\end{proof}

\begin{theorem}
$L \in \PSPACE$ iff for some $\M$ and some polynomial $p$
$$L=L(\M)= \EXACTSPACE{\M}{p}.$$
\end{theorem}

\begin{proof}
$\Rightarrow$. Consider $M'$ that accepts $L$ in space $p$. Then $L= \EXACTSPACE{\M}{p}$.

$\Leftarrow$. Assume $L= \EXACTSPACE{\M}{p}$. Then just simulate $\M$ during time $p$ to determine whether $w \in L$ or not. If you prefer, $\EXACTSPACE{\M}{p}$ is in $DSPACE(poly(p))$, and as $L= L(\M)= \EXACTSPACE{\M}{p}$ so is $L$.
\end{proof}
}

\olivierpasimportant{
subsection{To come: length}

Fix some distance $d$.
We define a notion of length.

\begin{itemize}
\item $L(\M)$ is the set of words accepted by $\M$.
\item $\LENGTH{\M}{n}$ is the  set of words accepted by $\M$ at a length $\le n$.

\item $\EXACTLENGTH{\M}{f}$ is the set of words accepted by $\M$ in length $f$: 
	$$\EXACTLENGTH{\M}{f} = \{w | w\in \LENGTH{\M}{f(|w|}\}.$$
\end{itemize}
}

\olivierpasimportant{\subsection{To come? time perturbations}

Olivier thinks it is possible to introduce the concept of (time) perturbed Turing machine $\M^{n}$.

\begin{itemize}
\item $\PERTURBEDT{\M}{n}$ corresponds to the associated language
\item $\PERTURBEDTIME{\M}{f}$ is the set of words accepted by $\M$ with time perturbation $f$: 
	$$\PERTURBEDTIME{\M}{f} = \{w | w\in \PERTURBEDT{\M}{f(|w|)}\}.$$
\end{itemize}
}

\olivierpasimportant{
subsection{To come: space perturbations}

Then Asarin and Bouajjani introduces the concept of (space) perturbed Turing machine $\M_{n}$.
	
\begin{itemize}
\item  $\PERTURBEDS{\M}{n}$ corresponds to the associated language.
\item $\PERTURBEDSPACE{\M}{f}$ is the set of words accepted by $\M$ with space perturbation $f$: 
	$$\PERTURBEDSPACE{\M}{f} = \{w | w\in \PERTURBEDS{\M}{f(|w|)}\}.$$
\end{itemize}

One of our result is:

\begin{theorem}
$L \in \PSPACE$ iff for some $\M$ and some polynomial $p$
$$L= L(\M)=\PERTURBEDSPACE{\M}{p}.$$
\end{theorem}
}

\olivier{Ce que avait gardé:}

\olivierpasimportant{ \begin{itemize}

%\item Et là, on est dans le papier de Asarin/Bouajjani, sauf qu'ils ne parlent pas de systèmes dynamiques généraux (à temps discret comme continu). Mais encore une fois, ca marcherait.
%
%\item Tout ca, ca parle de calculabilité.
%
%\item Mais si on suppose $M$ robuste, on sait que $P(M)$ décidable.
%
%\item Ce qu'on veut farie, c'est discuté de la complexté de $P(M)$ quand c'est décidable.
%
%\item Pour l'espace, finalement, pour que $P(M)$ soit dans $PSPACE$, il (faut et il) suffit que $$P(M)=P_{\omega}(M) = \cap_{n} Reach_{G_{n}}(cube_{n},cube_{n}')$$ aussi soit dans $\PSPACE$.
%
%Du coup, on regarde la complexité de $Reach_{G_{n}}(cube_{n},cube_{n}')$: c'est donné par la remarque de Manon, sur savitch. ca reste dans NPSPACE=PSPACE, si on a un nombre de cubes qui expose pas trop, = polynomial en $n$.
%
\item Du coup,  

\begin{equation}
\mbox{\begin{minipage}{10cm} pour que $P(M)$ soit dans $PSPACE$, 
 il  suffit que le graphe n'aie pas trop de sommets, et que $n$ reste bien pas trop énorme.
\end{minipage}}
\end{equation}

(pas trop, c'est polynomial essentiellement)
%% coupure commentaire:
\end{itemize} }

\olivierpasimportant{ \begin{itemize}

\item 
Et ca, finalement, c'est un résultat que je treouve super joli. On sait quand est que c'est dans $PSPACE$. Et c'est quoi un langage $PSPACE$. 

\item Et d'ailleurs le nombre de cubes, ca mesure un truc: la complexité en espace de la chose. Et changer la fa\c con de construire les cubes, ne change pas trop de choses, tant qu'on a bien les arguments plus haut. 

 \item Meme chose pour le temps: mais la propriété sur les graphes pour le temps, elle a avoir avec la longueur du chemin dans le graphe: un modèle de perturbation, ou on fait n'importe quoi après un certain temps, pas un certain espace (= certain n).

\end{itemize}
}

%%%%%%%%%%%%%%%%%%%%%%%%%%%%%%%%%%%%%%%%%%%%%%%%%%%
%%%%%%%%%%%%%%%%%%%%%%%%%%%%%%%%%%%%%%%%%%%%%%%%%%%
%%%%%%%%%%%%%%%%%%%%%%%%%%%%%%%%%%%%%%%%%%%%%%%%%%%
%%%%%%%%%%%%%%%%%%%%%%%%%%%%%%%%%%%%%%%%%%%%%%%%%%%
%%%%%%%%%%%%%%%%%%%%%%%%%%%%%%%%%%%%%%%%%%%%%%%%%%%
%%%%%%%%%%%%%%%%%%%%%%%%%%%%%%%%%%%%%%%%%%%%%%%%%%%
%%%%%%%%%%%%%%%%%%%%%%%%%%%%%%%%%%%%%%%%%%%%%%%%%%%

\end{document}